\documentclass{llncs}

\usepackage{amssymb,amsfonts,booktabs,tabularx}
\usepackage{url,xspace,soul,wrapfig}
\usepackage{graphicx,hyperref,paralist}
\usepackage{paralist,enumerate,soul,lineno}
\usepackage[font=small, hypcap=true]{subfig,caption}
\usepackage[textsize=footnotesize]{todonotes}
\usepackage{amsmath}

\makeatother

\graphicspath{{fig/}}

\newcommand{\bt}[1]{\ensuremath{bt(#1)}}
\newcommand{\dbt}[1]{\ensuremath{dbt(#1)}}
\newcommand{\tw}[1]{\ensuremath{\mathsf{\MakeUppercase{#1}}}}
\newcommand{\cn}[1]{\ensuremath{\mathsf{#1}}}

\newcommand{\myparagraph}[1]{\medskip\noindent\textbf{#1} }
\newtheorem{pattern}{Forbidden~Pattern}

\newcommand{\Ered}{\ensuremath{E_{\text{r}}} }
\newcommand{\Egreen}{\ensuremath{E_{\text{g}}} }
\newcommand{\Eblue}{\ensuremath{E_{\text{b}}} }
\newcommand{\EredT}{\ensuremath{T_{\text{r}}} }
\newcommand{\EredN}{\ensuremath{N_{\text{r}}} }
\newcommand{\Dred}{\ensuremath{V^*_{\text{r}}} }
\newcommand{\Dgreen}{\ensuremath{V^*_{\text{g}}} }
\newcommand{\Dblue}{\ensuremath{V^*_{\text{b}}} }

\newcommand{\cupdot}{\sqcup}

\title{On Dispersable Book Embeddings}

\author{Jawaherul~Md.~Alam\inst{1} \and 
Michael~A.~Bekos\inst{2} \and 
Martin~Gronemann\inst{3} \and
Michael~Kaufmann\inst{2} \and 
Sergey~Pupyrev\inst{1}}

\institute{
Dep. of Computer Science, University of Arizona, Tucson, USA
\\\email{\{jawaherul,spupyrev\}@gmail.com}
\and
Institut f{\"u}r Informatik, Universit{\"a}t T{\"u}bingen, T{\"u}bingen, Germany
\\\email{\{bekos,mk\}@informatik.uni-tuebingen.de}
\and
Institut f\"ur Informatik Universit\"at zu K\"oln, K\"oln, Germany
\\\email{gronemann@informatik.uni-koeln.de}
}
\begin{document}
\maketitle

\pagenumbering{arabic}

% ============================================================
\begin{abstract} 
In a \emph{dispersable book embedding}, the vertices of a given graph~$G$ must be ordered along a line $\ell$, called \emph{spine}, and the edges of $G$ must be drawn at different half-planes bounded by $\ell$, called \emph{pages of the book}, such that: 
\begin{inparaenum}[(i)]
\item no two edges of the same page cross, and
\item the graphs induced by the edges of each page are $1$-regular.
\end{inparaenum} 
The minimum number of pages needed by any dispersable book embedding of $G$ is referred to as the \emph{dispersable book thickness} $dbt(G)$ of $G$. Graph $G$ is called \emph{dispersable} if $\dbt{G} = \Delta(G) $ holds (note that $\Delta(G) \leq \dbt{G}$ always holds). 

Back in 1979, Bernhart and Kainen conjectured that any $k$-regular bipartite graph $G$ is dispersable, i.e., $\dbt{G}=k$. 
%Overbay, in her Ph.D.~thesis, observed that bipartiteness is an necessary condition in this conjecture, as any $k$-regular dispersable graph must be in fact bipartite. 
In this paper, we disprove this conjecture for the cases $k=3$ (with a computer-aided proof), and $k=4$ (with a purely combinatorial proof). In particular, we show that the Gray graph, which is $3$-regular and bipartite, has dispersable book thickness four, while the Folkman graph, which is $4$-regular and bipartite, has dispersable book thickness five. On the positive side, we prove that $3$-connected $3$-regular bipartite planar graphs are dispersable, and conjecture that this property holds, even if $3$-connectivity is relaxed.   
\end{abstract}
% ============================================================

% ============================================================
\section{Introduction}
\label{sec:introduction}
The book-embedding problem is a well studied problem in graph theory due to its numerous applications (see, e.g.,~\cite{CLR87,HLR92,Ros83,Tar72}) with early results dating back in early 1970s~\cite{Oll73}. The input in this problem is a graph $G$ and the task is to find a \emph{linear order} of the vertices of $G$ along a line $\ell$, called the \emph{spine of the book}, and an assignment of the edges of $G$ to different half-planes, called \emph{pages of the book}, delimited by the spine, such that no two edges of the same page cross (see Fig.~\ref{fig:dispresable-book-embedding} for an illustration). The minimum number of pages that are required by any book embedding of $G$ is commonly referred to as \emph{book thickness} (but also as \emph{stack number} or \emph{page number}) and is denoted by $\bt{G}$.

For planar input graphs, the literature is really rich. The most notable result is due to Yannakakis~\cite{Yan89}, who proved that the book thickness of a planar graph is at most four improving upon previous results~\cite{DBLP:conf/stoc/BussS84,Hea84}. Better upper bounds are only known for restricted subclasses, such as planar $3$-trees~\cite{Hea84} (which fit in books with three pages), subgraphs of planar Hamiltonian graphs~\cite{BK79}, $4$-connected planar graphs~\cite{NC08}, planar graphs without separating triangles~\cite{KOL07}, Halin graphs~\cite{CNP83}, bipartite planar graphs~\cite{DBLP:journals/dcg/FraysseixMP95}, planar $2$-trees~\cite{CLR87}, planar graphs of maximum degree~$3$ or $4$~\cite{Hea85,BGR14} (which fit in books with two pages), and outerplanar graphs~\cite{BK79} (which fit in single-page books). Note that, in general, the problem of testing, whether a maximal planar graph has book thickness two, is equivalent to determining whether it is Hamiltonian, and thus is NP-complete~\cite{Wig82}.

For non-planar input graphs, the literature is significantly limited. It is known that the book thickness of a complete $n$-vertex graph is $\Theta(n)$~\cite{BK79}, while sublinear book thickness have all graphs with subquadratic number of edges~\cite{Mal94a}, subquadratic genus~\cite{Mal94b} or sublinear treewidth~\cite{DW07}. The book thickness is known to be bounded only for bounded genus graphs~\cite{Mal94b} and, more generally, all minor-closed graph families~\cite{Bla03}. The reader is referred to~\cite{DBLP:journals/dmtcs/DujmovicW04} for a survey.

\begin{figure}[t]
	\centering
	\subfloat[\label{fig:dispresable-book-embedding}{}]
	{\includegraphics[page=1,scale=0.5]{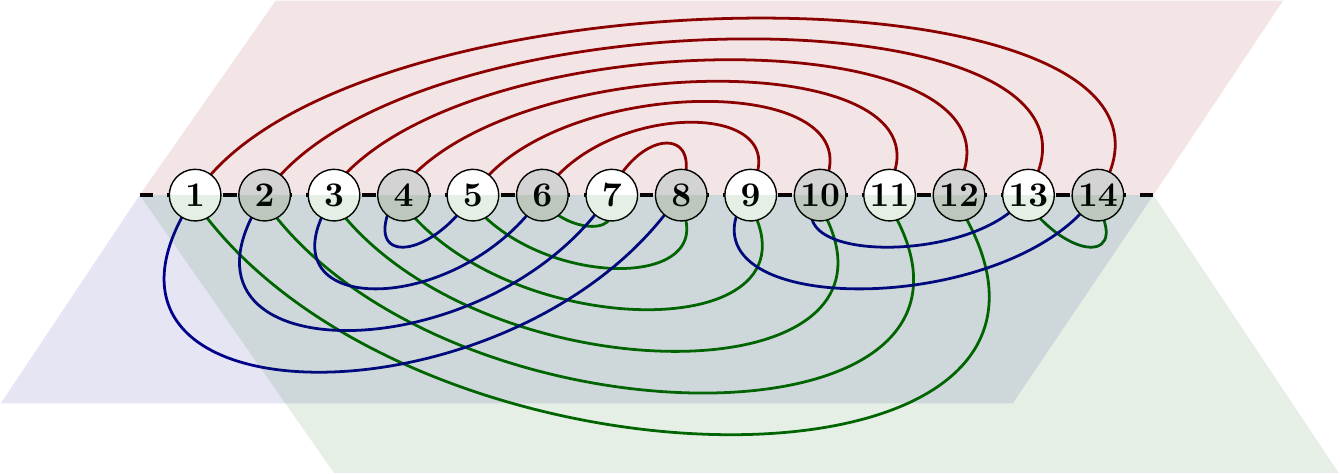}}
	\hfil
	\subfloat[\label{fig:circular-book-embedding}{}]
	{\includegraphics[page=2,scale=0.5]{heawood}}
	\caption{%
	(a)~A dispersable book embedding of the $3$-regular bipartite Heawood graph~\cite{Ger09} with $3$ pages~(taken from \cite{Kai11}), and
	(b)~an equivalent circular embedding with a $3$-edge-coloring, in which no two edges of the same color~cross.}
	\label{fig:heawood}
\end{figure}

In this paper, we focus on \emph{dispersable book embeddings}~\cite{BK79}, in which the subgraphs induced by the edges of each page are additionally required to be $1$-regular (i.e., matchings). The \emph{dispersable book thickness} of a graph $G$, denoted by $\dbt{G}$, is defined analogously to the book thickness as the minimum number of pages required by any dispersable book embedding of $G$. So, by definition $\dbt{G} \geq \Delta(G)$ holds, where $\Delta(G)$ is the maximum degree of $G$. Finally, a graph $G$ is called \emph{dispersable} if and only if $\dbt{G} = \Delta(G)$; see Fig.~\ref{fig:dispresable-book-embedding}. 

We note here that any book embedding with $k$ pages can be equivalently transformed into a circular embedding with a $k$-edge-coloring, in which no two edges of the same color cross, and vice versa~\cite{BK79,Hea84}. In the dispersable case, the graphs induced by the edges of the same color must additionally be $1$-regular; see Fig.~\ref{fig:circular-book-embedding}. We refer to the order, in which the vertices appear along the boundary of a circular embedding with $\Delta(G)$ colors (or, equivalently along the spine of a dispersable book emdedding with $\Delta(G)$ pages), if any, as \emph{dispersable~order}.

Dispersable book embeddings were first studied by Bernhart and Kainen~\cite{BK79}, who back in 1979 proved that the book thickness of the graph formed by the cartesian product of a dispersable bipartite graph $B$ and an arbitrary graph $H$ is upper bounded by the degree of $B$ plus the book thickness of $H$ (that is, $\bt{B \times H} \leq \bt{H} + \Delta(B)$), and posed the following conjecture (see also~\cite{Kai11}):
\begin{conjecture}[Bernhart and Kainen, 1979]\label{con:BK79}
Every $k$-regular bipartite graph $G$ is dispersable, that is, $\dbt{G}=k$.
\end{conjecture}

Clearly, Conjecture~\ref{con:BK79} holds for $k \leq 2$. As every $k$-regular bipartite graph admits a proper $k$-edge-coloring, Conjecture~\ref{con:BK79} implies that the dispersable book thickness of a regular bipartite graph equals its chromatic index. Overbay~\cite{Over98}, who continued the study of dispersable embeddings in her Ph.D.\ thesis, observed that not every proper $k$-edge coloring yields a dispersable book embedding and that bipartiteness is a necessary condition in the conjecture of Bernhart and Kainen. She also proved that several classes of graphs are dispersable; among them trees, binary cube graphs, and complete graphs. 

\myparagraph{Our contribution:} In Section~\ref{sec:4-regular}, we disprove Conjecture~\ref{con:BK79} for the case $k=4$, by showing, with a purely combinatorial proof, that the Folkman graph (see~Fig.~\ref{fig:folkman}), which is $4$-regular and bipartite, has dispersable book thickness five. In Section~\ref{sec:3-regular}, we first show how one can appropriately adjust a relatively recent SAT-formulation of the book embedding problem~\cite{BKZ15} for the dispersable case, and, using this formulation, we demonstrate that the Gray graph (see~Fig.~\ref{fig:gray}), which is $3$-regular and bipartite, has dispersable book thickness four (thus, disproving Conjecture~\ref{con:BK79} also for the case $k=3$). Note that, since both graphs are not planar, their book thickness is at least three. Figs.~\ref{fig:folkman-book-embedding} and~\ref{fig:gray-book-embedding} demonstrate that it is exactly three. In Section~\ref{sec:3-regular-planar}, we show that $3$-connected $3$-regular bipartite planar graphs are dispersable. Our findings lead to a number of interesting research directions, which we list in Section~\ref{sec:conclusions}, where we also conjecture that all (i.e., not necessarily $3$-connected) $3$-regular planar bipartite graphs are dispersable. 

\section{The Dispersable Book Thickness of the Folkman Graph}
\label{sec:4-regular}
In this section, we study the book thickness of the Folkman graph~\cite{FOLKMAN1967215}, which can be constructed in two steps starting from $K_5$ as follows. First, we replace every edge by a path of length two to obtain a bipartite graph (see Fig.~\ref{fig:folkman_2}). Then, we add for every vertex of the original $K_5$ a copy with the same neighborhood (see Fig.~\ref{fig:folkman_3}). The resulting graph is the Folkman graph, which is clearly $4$-regular and bipartite. We refer to a vertex of the original $K_5$ and to its copy as \emph{twin} vertices. The remaining vertices of the Folkman, i.e., the ones obtained from the paths, are referred to as \emph{connector vertices}. We denote the five pairs of twin vertices by $\tw{a_1}$, $\tw{a_2}$, $\tw{b_1}$, $\tw{b_2}$, $\tw{c_1}$, $\tw{c_2}$, $\tw{d_1}$, $\tw{d_2}$, $\tw{e_1}$, $\tw{e_2}$, and the ten connector vertices by $\cn{ab}$, $\cn{ac}$, $\cn{ad}$, $\cn{ae}$, $\cn{bc}$, $\cn{bd}$, $\cn{be}$, $\cn{cd}$, $\cn{ce}$, $\cn{de}$; see Fig.~\ref{fig:folkman_3}. 
%For example, twin $\tw{a_1}$ is adjacent to four connectors via edges $(\tw{a_1}, \cn{ab})$, $(\tw{a_1}, \cn{ac})$, $(\tw{a_1}, \cn{ae})$, and $(\tw{a_1}, \cn{ad})$; see Fig.~\ref{fig:folkman_3}.% for an illustration.

\begin{figure}[t]
	\centering
	\subfloat[\label{fig:folkman_1}{}]
	{\includegraphics[page=1,width=0.22\textwidth]{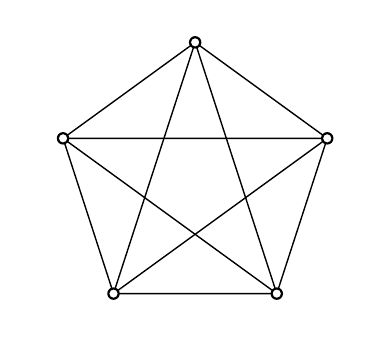}}
	\hfil
	\subfloat[\label{fig:folkman_2}{}]
	{\includegraphics[page=2,width=0.22\textwidth]{folkman}}
	\hfil
	\subfloat[\label{fig:folkman_3}{}]
	{\includegraphics[page=3,width=0.22\textwidth]{folkman}}
	\caption{Construction steps for the Folkman graph~\cite{FOLKMAN1967215}.}
	\label{fig:folkman}
\end{figure}
To prove that the dispersable book thickness of the Folkman graph is five, it suffices to prove that its dispersable book thickness cannot be four, and that it admits a dispersable book embedding with five pages. For the later, refer to Fig.~\ref{fig:folkman-dispersable} in Appendix~\ref{app:folkman}. For the former, we will assume for a contradiction that the Folkman graph admits a circular embedding with a $4$-edge-coloring, in which
\begin{inparaenum}[(i)]
\item no two edges of the same color cross, and
\item the graphs induced by the edges of the same color are $1$-regular.
\end{inparaenum}
Since by Property~(ii) adjacent edges must have different colors, we name them ``crossing'' such that we can use Property~(i) also for them. In the drawings, we use red, green, blue, and orange to indicate the four colors of the edges; black is used for an unknown (or not yet specified) color; see, e.g., Fig.~\ref{fig:lm22}. For any subset of at least three twin or connector vertices of the Folkman graph, say $\tw{A_1}$, $\cn{ab}$ and $\tw{B_2}$, we denote the clockwise order in which they appear along the boundary of the circular embedding by $(\dots \tw{A_1} \dots \cn{ab} \dots \tw{B_2} \dots)$. Every two vertices, say $\cn{ab}$ and $\tw{A_1}$, form two intervals, $[\cn{ab}, \tw{A_1}]$ and $[\tw{A_1}, \cn{ab}]$, in the clockwise order that correspond to the two arcs on the circle.

% ============================================================
\myparagraph{Useful lemmas:}
% ============================================================
In the following, we investigate properties of a dispersable book embedding with four pages of the Folkman graph. We start with a property that was first observed by Overbay~\cite{Over98} and later reproved by Hoske~\cite{Hos12}.

\begin{lemma}[Overbay~\cite{Over98}]
\label{lm:alternate}
For any regular bipartite graph, the vertices from both partitions are alternating in a dispersable order.
\end{lemma}
 
For the Folkman graph, Lemma~\ref{lm:alternate} implies that twin and connector vertices are alternating, i.e., for every pair of twins $\tw{a}$ and $\tw{b}$, interval $[\tw{a}, \tw{b}]$ contains a connector $\cn{x}$, that is, the order is always $(\dots \tw{a} \dots \cn{x} \dots \tw{b} \dots \cn{y} \dots)$. The next lemma strengthen the claim by describing an interval between twins $\tw{a_1}$ and $\tw{a_2}$.

\begin{lemma}
\label{lm:22}
Let $\tw{a_1}$, $\tw{a_2}$ be a pair of twins and $[\tw{a_1}, \tw{a_2}]$, $[\tw{a_2}, \tw{a_1}]$ be two intervals defined by the twins. Then one of the following holds:
\begin{itemize}[-]
\item \label{c:22:1} one of the intervals contains exactly one connector vertex corresponding to the twins, and another one contains all other connectors, that is, the order is $(\dots \tw{a_1}~\cn{ax}~\tw{a_2} \dots \cn{ay} \dots \cn{au} \dots \cn{av} \dots)$;

\item \label{c:22:2} both intervals contain two connectors corresponding to the twins (and possibly other connectors), that is, the order is
$(\dots \tw{a_1} \dots \cn{ax} \dots \cn{ay} \dots \tw{a_2} \dots \cn{au} \dots \cn{av} \dots)$.
\end{itemize}
\end{lemma}
\begin{proof}
The twins $\tw{a_1}$ and $\tw{a_2}$ have four connectors, $\cn{ax}$, $\cn{ay}$, $\cn{au}$, and $\cn{av}$. Let us first show that it is impossible for one interval to contain all four connectors. Assume for a contradiction that interval $[\tw{a_1}, \tw{a_2}]$ contains all the connectors; the order is $(\dots \tw{a_1} \dots \cn{ax} \dots \cn{ay} \dots  \cn{au} \dots \cn{av} \dots \tw{a_2} \dots )$. It is easy to see that five edges, $(\tw{a_1}, \cn{ax})$, $(\tw{a_1}, \cn{ay})$, $(\tw{a_1}, \cn{au})$, $(\tw{a_1}, \cn{av})$, $(\tw{a_2}, \cn{ax})$, are pairwise crossing. Thus, they need five distinct colors, which is impossible in a dispersable order of the Folkman graph.

\begin{figure}[t]
	\centering
	\subfloat[\label{fig:lm22:0}{}]{
	\includegraphics[page=1,width=0.19\textwidth]{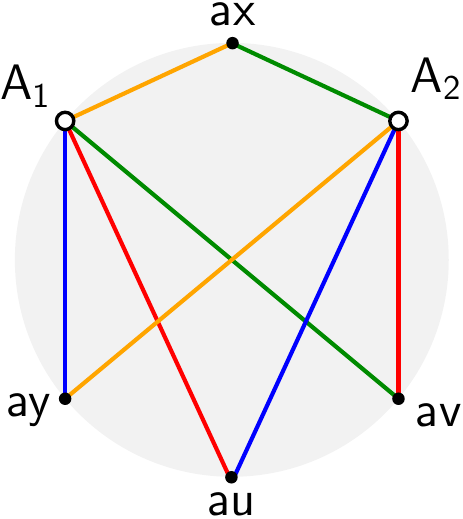}}
	\hfill
	\subfloat[\label{fig:lm22:1}{}]{
	\includegraphics[page=2,width=0.19\textwidth]{lm22}}
	\hfill
	\subfloat[\label{fig:lm22:2}{}]{
	\includegraphics[page=3,width=0.19\textwidth]{lm22}}
	\hfill
	\subfloat[\label{fig:lm22:3}{}]{
	\includegraphics[page=4,width=0.19\textwidth]{lm22}}
	\caption{Illustration for the proof of Lemma~\ref{lm:22}.
	%The four edge-colors are orange, blue, green, and red.
	%Black dashed edges indicate a contradiction.
	}
	\label{fig:lm22}
\end{figure}

To complete the proof, we find a contradiction for the case in which $[\tw{a_1}, \tw{a_2}]$ contains only one $\tw{a}$'s connector and some other connectors not adjacent to $\tw{a_1}$ and $\tw{a_2}$. By Lemma~\ref{lm:alternate}, $[\tw{a_1}, \tw{a_2}]$ also contains other twin vertices. Denote $\tw{a}$'s connector in $[\tw{a_1}, \tw{a_2}]$ by $\cn{ax}$ and other three connectors by $\cn{ay}$, $\cn{au}$, $\cn{av}$. Notice that the coloring of the eight $\tw{A}$'s edges is unique (up to color shift); see Fig.~\ref{fig:lm22:0}. We distinguish the three cases depending on the twins lying in $[\tw{a_1}, \tw{a_2}]$:
\begin{itemize}[-]
\item Interval $[\tw{a_1}, \tw{a_2}]$ contains exactly one twin vertex, say $\tw{B_1}$, that is, the order is $(\dots \tw{a_1} \cdot \tw{b_1} \cdot \tw{a_2} \dots)$ with two intermediate connectors (by Lemma~\ref{lm:alternate}). It is easy to see that connector $\cn{ab} \notin [\tw{a_2}, \tw{a_1}]$, as otherwise edge $(\tw{b_1}, \cn{ab})$ cannot be colored. Hence, $\cn{ab} \in [\tw{a_1}, \tw{a_2}]$; see Fig.~\ref{fig:lm22:1}. Let~$\cn{r}$ be the second connector in $[\tw{a_1}, \tw{a_2}]$. Then, the green edge adjacent to $\tw{b_1}$ must be $(\tw{b_1}, \cn{r})$, and therefore edge $(\tw{b_2}, \cn{r})$ cannot be colored; a contradiction.

\item Interval $[\tw{a_1}, \tw{a_2}]$ contains exactly two same twins, say $\tw{B_1},\tw{B_2}$, that is, the order is $(\dots \tw{a_1} \cdot \tw{b_1} \cdot \tw{b_2} \cdot \tw{a_2} \dots)$ with three intermediate connectors (by Lemma~\ref{lm:alternate}); see Fig.~\ref{fig:lm22:2}. If there is a connector $\cn{r}$ of $\tw{B}$ in the interval $[\tw{a_2}, \tw{a_1}]$, such that $\cn{r} \neq \cn{ab}$, then both edges $(\tw{b_1},\cn{r})$ and $(\tw{b_2},\cn{r})$ cannot be colored without forming five edges that pairwise cross, which is a contradiction. We claim that there is such connector. 
If $\cn{ab} \in [\tw{a_2}, \tw{a_1}]$, then, since $\cn{ab} \neq \cn{ax}$, at most two connectors of $\tw{b}$ can be on the interval $[\tw{a_1}, \tw{a_2}]$, which implies besides $\cn{ab}$ there is one more connector of $\tw{b}$ in $[\tw{a_2}, \tw{a_1}]$, as desired. Otherwise, $\cn{ab} \in [\tw{a_1}, \tw{a_2}]$. In this case, at most three connectors of $\tw{b}$ can be on the interval $[\tw{a_1}, \tw{a_2}]$, which implies there is connector $\cn{r}$ of $\tw{b}$ in $[\tw{a_2}, \tw{a_1}]$, such that $\cn{r} \neq \cn{ab}$.

\item Interval $[\tw{a_1}, \tw{a_2}]$ contains two or more different twins, say $\tw{B_1},\tw{C_1}$, that is, the order is $(\dots \tw{a_1} \dots \tw{b_1} \dots \tw{c_1} \dots \tw{a_2} \dots)$; see Fig.~\ref{fig:lm22:3}. One of the connectors $\cn{ab}$ and $\cn{ac}$ is in $[\tw{a_2}, \tw{a_1}]$, since we assumed that $\tw{a_1}$ and $\tw{a_2}$ have only one connector on the opposite interval $[\tw{a_1}, \tw{a_2}]$. W.l.o.g.\ assume connector $\cn{ac}\in[\tw{a_2}, \tw{a_1}]$. Then, the edge $(\tw{c_1}, \cn{ac})$ cannot be colored; a contradiction.\qed
\end{itemize}
\end{proof}

Denote the number of $\tw{a}$'s connectors in $[x, y]$ by~$\delta_{\tw{a}}(x, y)$. Lemma~\ref{lm:22} defines two possible configurations for a pair of twins, $\tw{a_1}$ and $\tw{a_2}$. The first one, which we call \emph{1-3 configuration}, is when $\delta_{\tw{a}}(\tw{a_1}, \tw{a_2}) = 1$ and $\delta_{\tw{a}}(\tw{a_2}, \tw{a_1}) = 3$, that is, the first interval contains one connector and another interval contains three connectors. In that case, the twins have to lie next to each other in the order (that is, there is no other twins in between); we call such twins \emph{close}. In the second configuration, called \emph{2-2 configuration}, $\delta_{\tw{a}}(\tw{a_1}, \tw{a_2}) = \delta_{\tw{a}}(\tw{a_2}, \tw{a_1}) = 2$ holds. Here, the twins are called \emph{far} (as there is at least one other twin in between).

%Can we say something about two pairs of twins in a dispersable order? 
\noindent The next two lemmas describe properties of pairs of twins based on whether~they alternate along the spine (\emph{crossing twin-pairs}) or not (\emph{non-crossing twin-pairs}).

\begin{lemma}[non-crossing twin pairs]
\label{lm:noncross}
Let $\tw{\tw{a_1}}$, $\tw{\tw{a_2}}$ and $\tw{b_1}$, $\tw{b_2}$ be two pairs of \emph{non-crossing twins}, that is, the order is $(\dots \tw{a_1} \dots \tw{a_2} \dots \tw{b_1} \dots \tw{b_2} \dots)$. For connector $\cn{ab}$, one of the following holds:

\begin{enumerate}[i.]
\item \label{c:noncross-1} $\cn{ab}$ is in $[\tw{a_2}, \tw{b_1}]$, that is, $(\dots \tw{a_1} \dots \tw{a_2} \dots \cn{ab} \dots \tw{b_1} \dots \tw{b_2} \dots)$;

\item \label{c:noncross-2} $\cn{ab}$ is in $[\tw{b_2}, \tw{a_1}]$, that is, $(\dots \tw{a_1} \dots \tw{a_2} \dots \tw{b_1} \dots \tw{b_2} \dots \cn{ab} \dots)$;

\item \label{c:noncross-3} $\tw{a_1}$ and $\tw{a_2}$ are close twins and the four twins are separated by $\tw{a}$'s connectors, that is, the order is $(\dots \tw{a_1}~\cn{ab}~\tw{a_2} \dots \cn{ax} \dots \tw{b_1} \dots \cn{ay} \dots \tw{b_2} \dots \cn{az} \dots)$;

\item \label{c:noncross-4} $\tw{b_1}$ and $\tw{b_2}$ are close twins and the four twins are separated by $\tw{b}$'s connectors, that is, the order is $(\dots \tw{b_1}~\cn{ab}~\tw{b_2} \dots \cn{bx} \dots \tw{a_1} \dots \cn{by} \dots \tw{a_2} \dots \cn{bz} \dots)$.
\end{enumerate}
\end{lemma}
\begin{proof}
If $\cn{ab}$ is in $[\tw{a_2}, \tw{b_1}]$ or $[\tw{b_2}, \tw{a_1}]$, then the lemma holds. So, let $\cn{ab}$ be in $[\tw{a_1},\tw{a_2}]$. If $\tw{a_1}$ and $\tw{a_2}$ are far, then by Lemma~\ref{lm:22} there is~another of $\tw{a}$'s connectors, say $\cn{ax}$, in $[\tw{a_1},\tw{a_2}]$. One of the edges $(\tw{a_1}, \cn{ax})$, $(\tw{a_2}, \cn{ax})$ cannot be colored; see Fig.~\ref{fig:lmaabb1}. Thus, $\cn{ab}$ is the only of $\tw{a}$'s connector on $[\tw{a_1},\tw{a_2}]$, and $\tw{a_1},\tw{a_2}$ are close.

\begin{figure}[t]
	\centering
	\subfloat[\label{fig:lmaabb1}{}]{
	\includegraphics[page=1, width=0.19\textwidth]{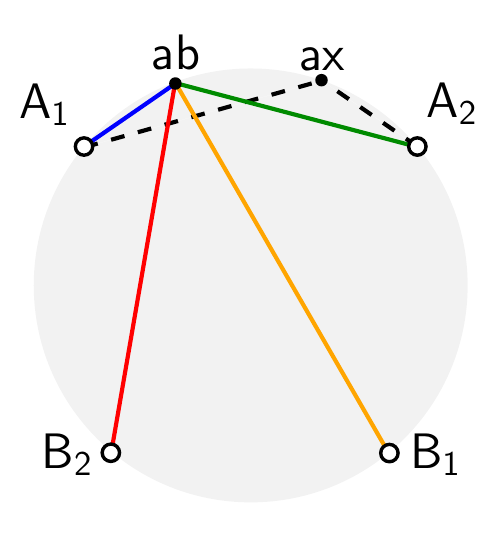}}
	\hfil
	\subfloat[\label{fig:lmaabb2}{}]{
	\includegraphics[page=3, width=0.19\textwidth]{lmnoncross}}
	\hfil
	\subfloat[\label{fig:lmaabb3}{}]{
	\includegraphics[page=2, width=0.19\textwidth]{lmnoncross}}
	\hfil
	\subfloat[\label{fig:lmaabb4}{}]{
	\includegraphics[page=4, width=0.19\textwidth]{lmnoncross}}
	\caption{Illustration for the proof of Lemma~\ref{lm:noncross}. 
	%Black dashed edges indicate a contradiction.
	}
	\label{fig:lmnoncross_1}
\end{figure}

Twins $\tw{b_1}$ and $\tw{b_2}$ define three sub-intervals on $[\tw{a_2}, \tw{a_1}]$. If two of $\tw{a}$'s connectors, say $\cn{ax}$ and $\cn{ay}$, belong to the leftmost sub-interval, then $(\tw{a_2}, \cn{ax})$, $(\tw{a_2},\cn{ay})$, $(\tw{b_1}, \cn{ab})$ and $(\tw{b_2}, \cn{ab})$ pairwise cross, which implies that all must have different colors; see Fig.~\ref{fig:lmaabb2}. Thus, $(\tw{a_2}, \cn{ab})$ needs a fifth color; a contradiction. Symmetric is the case, in which two of $\tw{a}$'s connectors belong to the rightmost sub-interval. Finally, if two of $\tw{a}$'s connectors, say $\cn{ax}$ and $\cn{ay}$, are on the central sub-interval, then by symmetry we may assume that the fourth of $\tw{a}$'s connectors, say $\cn{az}$, belongs either to $[\tw{a_2}, \tw{B_1}]$ or to $[\tw{b_1}, \tw{b_2}]$. In both cases, edge $(\tw{b_2}, \cn{ab})$ crosses $(\tw{a_1}, \cn{ax})$, $(\tw{a_1}, \cn{ay})$ and $(\tw{a_1}, \cn{az})$, which implies that all must have different colors; see Figs.~\ref{fig:lmaabb3} and~\ref{fig:lmaabb4}. Thus, $(\tw{a_1}, \cn{ab})$ needs a fifth color; a contradiction. We conclude that the three intervals contain one connector each, as in the claim.\qed
\end{proof}

%\noindent A simple corollary of Lemma~\ref{lm:noncross} is the following.

\begin{corollary}
\label{lm:close}
Let the order be $(\dots \tw{a_1}~\cn{ab}~\tw{a_2} \dots)$. Then $\tw{b_1}$ and $\tw{b_2}$ are far twins.
\end{corollary}
\begin{proof}
By Lemma~\ref{lm:noncross}.\ref{c:noncross-3}, $\tw{b_1}$ and $\tw{b_2}$ are separated by $\tw{a}$'s connectors, that is, the order is $(\dots \tw{a_1}~\cn{ab}~\tw{a_2} \dots \cn{ax} \dots \tw{b_1} \dots \cn{ay} \dots \tw{b_2} \dots \cn{az} \dots)$. If $\tw{b_1}$ and $\tw{b_2}$ are close, then the twins are in 1-3 configuration. Thus, the only connector between them is a $\tw{b}$'s connector, a contradiction.\qed
\end{proof}

%\begin{corollary}
%\label{lm:tooclose}
%Two pairs $\tw{a_1},\tw{a_2}$ and $\tw{b_1},\tw{b_2}$ of close twins are not separated by their common connector $\cn{ab}$.
%\end{corollary}

\begin{lemma}[crossing twin pairs]
\label{lm:cross}
Let $\tw{\tw{a_1}}$, $\tw{\tw{a_2}}$ and $\tw{b_1}$, $\tw{b_2}$, be two pairs of crossing twins, that is, the order is $(\dots \tw{a_1} \dots \tw{b_1} \dots \tw{a_2} \dots \tw{b_2} \dots)$. Then one of the following holds:
\begin{enumerate}[i.]
\item \label{c:cross-1} $\delta_{\tw{a}}(\tw{a_1}, \tw{b_1}) = \delta_{\tw{a}}(\tw{b_1}, \tw{a_2}) = \delta_{\tw{a}}(\tw{a_2}, \tw{b_2}) = \delta_{\tw{a}}(\tw{b_2}, \tw{a_1}) = 1$;
\item \label{c:cross-2} $\delta_{\tw{a}}(\tw{a_1}, \tw{b_1}) = \delta_{\tw{a}}(\tw{a_2}, \tw{b_2}) = 2$ and $\delta_{\tw{a}}(\tw{b_1}, \tw{a_2}) = \delta_{\tw{a}}(\tw{b_2}, \tw{a_1}) = 0$;
\item \label{c:cross-3} $\delta_{\tw{a}}(\tw{a_1}, \tw{b_1}) = \delta_{\tw{a}}(\tw{a_2}, \tw{b_2}) = 0$ and $\delta_{\tw{a}}(\tw{b_1}, \tw{a_2}) = \delta_{\tw{a}}(\tw{b_2}, \tw{a_1}) = 2$.
\end{enumerate}
In particular, $\tw{a}$'s connectors are in the 2-2 configuration with respect to $\tw{b}$'s twins, that is, $\delta_{\tw{a}}(\tw{b_1}, \tw{b_2}) = \delta_{\tw{a}}(\tw{b_2}, \tw{b_1}) = 2$.
\end{lemma}
\begin{proof}
Besides the three cases described in the lemma, we will exclude all the remaining cases, which are as follows:
\begin{itemize}[-]
\item There is an interval $I$ which is either $[\tw{a_1}, \tw{b_1}]$, or $[\tw{b_1}, \tw{a_2}]$ or $[\tw{a_2}, \tw{b_2}]$ or $[ \tw{b_2}, \tw{a_1}]$ with $\delta_{\tw{a}}(I) \geq 3$,
\item $\delta_{\tw{a}}(\tw{a_1}, \tw{b_1}) = \delta_{\tw{a}}(\tw{b_2}, \tw{a_1}) = 2$ and $\delta_{\tw{a}}(\tw{B_1}, \tw{a_2}) = \delta_{\tw{a}}(\tw{a_2}, \tw{b_2}) = 0$, and
\item $\delta_{\tw{a}}(\tw{b_1}, \tw{a_2}) = \delta_{\tw{a}}(\tw{a_2}, \tw{b_2}) = 2$ and $\delta_{\tw{a}}(\tw{a_1}, \tw{b_1}) = \delta_{\tw{a}}(\tw{b_2}, \tw{a_1}) = 0$.
\end{itemize}

We start with the first case. Note that by Lemma~\ref{lm:22}, $\delta_{\tw{a}}(I)$ cannot be $4$. Let w.l.o.g.\ $I=[\tw{a_1}, \tw{b_1}]$ and assume for a contradiction that $\delta_{\tw{a}}(I) = 3 $. Then, by symmetry, we may assume $\delta_{\tw{a}}(\tw{a_2}, \tw{b_2}) + \delta_{\tw{a}}(\tw{b_2}, \tw{a_1}) = 1$. It is easy to see that the three colors of the edges from $\tw{a_2}$ to the three of the $\tw{a}$'s connectors in $[\tw{a_1}, \tw{b_1}]$, uniquely determine the four colors for the edges from $\tw{a_1}$ to the $\tw{a}$'s connectors; see e.g., Fig.~\ref{fig:lmabab1}. Then, we check the four possible locations for connector $\cn{ab}$. For each case, we try to insert the two edges from $\cn{ab}$ to $\tw{b_1}$ and $\tw{b_2}$, and immediately achieve a contradiction; see Figs.~\ref{fig:lmabab1}-\ref{fig:lmabab4} for an illustration.

\begin{figure}[t]
	\centering
	\subfloat[\label{fig:lmabab1}{}]{
	\includegraphics[page=1, width=0.19\textwidth]{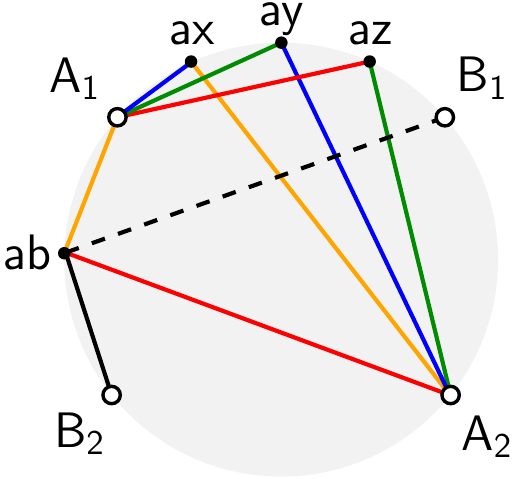}}
	\hfill
	\subfloat[\label{fig:lmabab2}{}]{
	\includegraphics[page=2, width=0.19\textwidth]{lmncross}}
	\hfill
	\subfloat[\label{fig:lmabab3}{}]{
	\includegraphics[page=3, width=0.19\textwidth]{lmncross}}
	\hfill
	\subfloat[\label{fig:lmabab4}{}]{
	\includegraphics[page=4, width=0.19\textwidth]{lmncross}}
	\hfill
	\subfloat[\label{fig:lmabab5}{}]{
	\includegraphics[page=5, width=0.19\textwidth]{lmncross}}
	\caption{Illustration for the proof of Lemma~\ref{lm:cross}.
	% Black dashed edges indicate a contradiction.
	}
	\label{fig:lmnoncross_2}
\end{figure}

For the second case, assume to the contrary that $\delta_{\tw{a}}(\tw{a_1}, \tw{b_1}) = \delta_{\tw{a}}(\tw{b_2}, \tw{a_1}) = 2$ which directly implies that $\delta_{\tw{a}}(\tw{b_1}, \tw{a_2}) = \delta_{\tw{a}}(\tw{a_2}, \tw{b_2}) = 0$. By symmetry, we may assume that connector $\cn{ab}$ is in $[\tw{a_1}, \tw{b_1}]$, and that it appears before the second of $\tw{a}$'s connectors, say $\cn{ax}$, in $[\tw{a_1}, \tw{b_1}]$, when moving along $[\tw{a_1}, \tw{b_1}]$ from $\tw{a_1}$ to $\tw{b_1}$; see Fig.~\ref{fig:lmabab5}. Since the edges from connector $\cn{ab}$ towards $\tw{a_1}$, $\tw{a_2}$, $\tw{b_1}$ and $\tw{b_2}$ use all colors, edge $(\tw{a_1},\cn{ax})$ cannot be colored; a contradiction. As the third case is symmetric to the second, the lemma follows.\qed
\end{proof}

\myparagraph{Case analysis:} We have now introduced the tools we need, and we proceed to we analyse several \emph{forbidden patterns}, that is, subsequences of twins, that cannot occur in a dispersable order of the Folkman graph.

\begin{figure}[t]
	\centering
	\subfloat[\label{fig:case3b1}{}]{
	\includegraphics[page=1, width=0.19\textwidth]{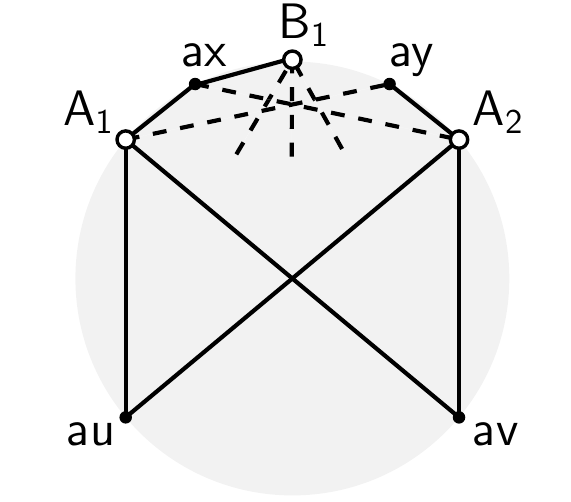}}
	\hfil
	\subfloat[\label{fig:case4a1}{}]{
	\includegraphics[page=2, width=0.19\textwidth]{case34}}
	\hfil
	\subfloat[\label{fig:case4a2}{}]{
	\includegraphics[page=3, width=0.19\textwidth]{case34}}
	\caption{Illustration for the proof of
	(a)~Forbidden Pattern~\ref{case3}, 
	(b)-(c)~Forbidden Pattern~\ref{case4a}.
	%Black dashed edges indicate a contradiction.
	}
	\label{fig:case34}
\end{figure}

\begin{pattern}[$\boldmath \dots \tw{a_1} \cdot \tw{b_1} \cdot \tw{a_2} \dots$]
\label{case3}
Between any twin pair, there is not exactly one single twin vertex.
\end{pattern}
\begin{proof}
Assume for a contradiction that there exists a pair of twins, $\tw{a_1},\tw{a_2}$, with exactly one twin vertex, $\tw{b_1}$, between them in $[\tw{a_1},\tw{a_2}]$. By Lemma~\ref{lm:22}, there are two connectors of $\tw{a}$ in $[\tw{a_1}, \tw{a_2}]$, call them $\cn{ax}$ and $\cn{ay}$, and two connectors of $\tw{a}$ in $[\tw{a_2}, \tw{a_1}]$, called them $\cn{au}$ and $\cn{av}$; see Fig.~\ref{fig:case3b1}. Twin $\tw{b_1}$ has four adjacent connectors, and only one of them can be a connector of $\tw{a}$, namely $\cn{ab}$. Hence, the three edges to the remaining connectors cross both $(\tw{A_1}, \cn{ay})$ and $(\tw{A_2}, \cn{ax})$; a contradiction.\qed
\end{proof}
\begin{pattern}[$\boldmath \dots \tw{a_1} \cdot \tw{b_1} \cdot \tw{b_2} \cdot \tw{a_2} \dots$]
\label{case4a}
Between any twin pair, there are not exactly two same twin vertices.
\end{pattern}
\begin{proof}
Assume to the contrary that there exists a pair of twins, $\tw{a_1},\tw{a_2}$, with exactly two same twin vertices, $\tw{b_1}, \tw{b_2}$, between them in $[\tw{a_1},\tw{a_2}]$. By Lemma~\ref{lm:alternate}, we assume that the order is $(\dots \tw{a_1}~\cn{x}~\tw{b_1}~\cn{y}~\tw{b_2}~\cn{z}~\tw{a_2} \dots)$, where $\cn{x},\cn{y},\cn{z}$ are connectors. By Lemma~\ref{lm:22}, two of them are connectors of $\tw{A}$, including $\cn{ab}$. If connector $\cn{ab}$ were $\cn{y}$, then by Lemma~\ref{lm:noncross}.\ref{c:noncross-4} both $\cn{x}$ and $\cn{z}$ would have been connectors of $\tw{b}$, contrading the fact that two of $\cn{x},\cn{y},\cn{z}$ are~connectors of $\tw{A}$. Hence, $\cn{ab}$ is not $\cn{y}$. It follows that there exist two $\tw{b}$'s connectors in $[\tw{a_2}, \tw{a_1}]$, call them $\cn{bu}$ and $\cn{bv}$; see Fig.~\ref{fig:case4a1}. Now, it is easy to see that edges $(\tw{b_1}, \cn{bv})$, $(\tw{b_2}, \cn{bu})$, $(\tw{b_2}, \cn{bv})$, $(\tw{a_1}, \cn{z})$, and $(\tw{a_2}, \cn{x})$ pairwise cross; see Fig.~\ref{fig:case4a2}. So, they need five colors; a contradiction.\qed
\end{proof}
\begin{pattern}[$\boldmath \dots \tw{a_1} \cdot \tw{b_1} \cdot \tw{c_1} \cdot \tw{a_2} \dots$]
\label{case4b}
Between any twin pair, there are not exactly two different twin vertices
\end{pattern}
\begin{proof}
Assume for a contradiction that there exists a pair of twins, $\tw{a_1},\tw{a_2}$, with exactly two different twin vertices, $\tw{b_1}, \tw{c_1}$, between them in $[\tw{a_1},\tw{a_2}]$. By Lemma~\ref{lm:alternate}, we assume that the order is $(\dots \tw{a_1}~\cn{x}~\tw{b_1}~\cn{y}~\tw{c_1}~\cn{z}~\tw{a_2} \dots)$, where $\cn{x},\cn{y},\cn{z}$ are connectors. By Lemma~\ref{lm:22}, twins $\tw{a_1}$ and $\tw{a_2}$ have two connectors in $[\tw{a_1}, \tw{a_2}]$ and two connectors in $[\tw{a_2}, \tw{a_1}]$. By Lemma~\ref{lm:noncross} applied first for twins $\tw{a}$ and $\tw{d}$, and then for $\tw{a}$ and $\tw{e}$, we conclude that $\cn{ad}$ and $\cn{ae}$ are on $[\tw{a_2}, \tw{a_1}]$, while $\cn{ab}$ and $\cn{ac}$ are on $[\tw{a_1}, \tw{a_2}]$. By symmetry, we may assume that  $\cn{ae}$ appears before $\cn{ad}$ on $[\tw{a_2}, \tw{a_1}]$.

\begin{figure}[t]
	\centering
	\subfloat[\label{fig:case4b1}{}]{
	\includegraphics[page=1, width=0.19\textwidth]{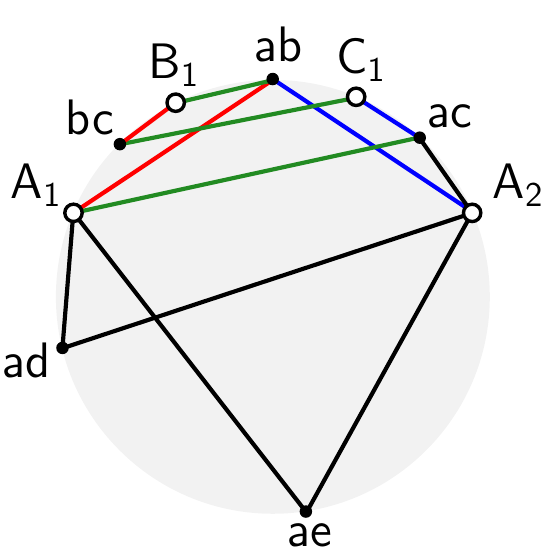}}
	\hfil
	\subfloat[\label{fig:case4b2}{}]{
	\includegraphics[page=2, width=0.19\textwidth]{case4}}
	\hfil
	\subfloat[\label{fig:case4b3}{}]{
	\includegraphics[page=3, width=0.19\textwidth]{case4}}
	\caption{Illustration of the case~$\cn{bc}=\cn{x} \in [\tw{a_1}, \tw{b_1}]$.
	%Black indicate unknown (not yet determined) edge color.
	}
\end{figure}

Since $\tw{a}$'s and $\tw{b}$'s twins cross, by Lemma~\ref{lm:cross} we obtain that $\delta_{\tw{b}}(\tw{a_1}, \tw{a_2}) = 2$ and there exist two $\tw{b}$'s connectors in $[\tw{a_1}, \tw{a_2}]$. Symmetrically, there are two $\tw{c}$'s connectors in $[\tw{a_1}, \tw{a_2}]$. Hence, connector $\cn{bc}$ is in $[\tw{a_1}, \tw{a_2}]$. By symmetry, we consider two cases: either $\cn{bc}=\cn{x} \in [\tw{a_1}, \tw{b_1}]$ or $\cn{bc}=\cn{y} \in [\tw{b_1}, \tw{c_1}]$.
\begin{itemize}[-]
\item Assume first that $\cn{bc}=\cn{x} \in [\tw{a_1}, \tw{b_1}]$. Since there should exist an edge $(\tw{b_1}, \cn{y})$ with the same color as $(\tw{c_1}, \cn{bc})$, it follows that $\cn{y} = \cn{ab}$ and $\cn{z} = \cn{ac}$; see Fig.~\ref{fig:case4b1}, in which we have also assumed a coloring started w.l.o.g.\ with green $(\tw{a_1}, \cn{ac})$, red $(\tw{a_1}, \cn{ab})$, and blue $(\tw{a_2}, \cn{ab})$.
Since $\delta_{\tw{a}}(\tw{b_1}, \tw{a_2}) = 2$, by Lemma~\ref{lm:cross}.\ref{c:cross-3} applied for $\tw{a}$ and $\tw{b}$, we obtain that $\delta_{\tw{a}}(\tw{a_2}, \tw{b_2}) = 0$. So, $\tw{b_2} \in [\tw{a_2}, \cn{ae}]$. Since $\delta_{\tw{a}}(\tw{a_1}, \tw{c_1}) = \delta_{\tw{a}}(\tw{c_1}, \tw{a_2}) = 1$, by Lemma~\ref{lm:cross}.\ref{c:cross-1} applied for $\tw{a}$ and $\tw{c}$, we have that~$\tw{c_2} \in [\cn{ae}, \cn{ad}]$. Next we color the edges; see Fig.~\ref{fig:case4b2}. Edge $(\tw{b_2}, \cn{ab})$ is orange (the  remaining color for $\cn{ab}$), and thus, $(\tw{c_2}, \cn{bc})$ and $(\tw{a_2}, \cn{ac})$ are orange (otherwise two orange edges cross), while $(\tw{b_2}, \cn{bc})$ is blue and $(\tw{c_2}, \cn{ac})$ is red. Hence, $(\tw{a_2}, \cn{ad})$ is green, $(\tw{a_1}, \cn{ae})$ is blue, $(\tw{a_1}, \cn{ad})$ is orange, and $(\tw{a_2}, \cn{ae})$~is~red.

Since $\tw{d_1}$ and $\tw{d_2}$ must be connected to $\cn{ad}$ via blue and red edges, $\tw{d_1}$ and $\tw{d_2}$ are in $[\cn{ae}, \tw{a_1}]$. The orange edge from $\tw{c_1}$ must end at a connector $\cn{cw} \notin \{\cn{ac},\cn{bc}\}$ in $[\tw{a_2},\tw{b_2}]$, and $\cn{cx}$ must have a green edge to $\tw{c_2}$. Since $\cn{cw}$ cannot be connected to $[\cn{ae}, \tw{a_1}]$, where $\tw{d_1}$ and $\tw{d_2}$ reside (as it would cross blue edge $(\tw{a_1}, \cn{ae})$ and red edge $(\tw{a_2}, \cn{ae})$), $\cn{cw} \neq \cn{cd}$ holds. Thus $\cn{cw} = \cn{ce} \in [\tw{a_2},\tw{b_2}]$.

Since both $\tw{e_1}$ and $\tw{e_2}$ have to be connected to $\cn{ce}$  via blue and red edges,~$\tw{e_1}$ and $\tw{e_2}$ are in $[\tw{a_2},\cn{ae}]$; see Fig.~\ref{fig:case4b3}. In particular, one of them is in $[\tw{a_2}, \tw{b_2}]$ (because it must be connected to $\cn{ce}$ with a blue edge), while the other one is in $[\tw{b_2}, \cn{ae}]$ (because it must be connected to $\cn{ae}$ with an orange edge). Since both $\tw{d_1}$ and $\tw{d_2}$ are in $[\cn{ae}, \tw{a_1}]$, we conclude that $\tw{e_1}, \tw{b_2}, \tw{e_2}$ form Forbidden Pattern~\ref{case3} $(\dots \tw{e_1} \cdot \tw{b_2} \cdot \tw{e_2} \dots)$, which is not possible.

\item Assume now that $\cn{bc}=\cn{y} \in [\tw{b_1}, \tw{c_1}]$. Note that $\cn{x} \neq \cn{ac}$, as otherwise $\tw{b_1}$ needs to have two edges in $[\cn{ac}, \tw{c_1}]$ (one with the color of $(\cn{ac}, \tw{a_2})$ and one with the color of $(\cn{ac}, \tw{c_1})$), which is impossible; see Fig.~\ref{fig:case4b4}. Since $\cn{ab},\cn{ac} \in [\tw{a_1}, \tw{a_2}]$, it follows that $\cn{x} = \cn{ab}$ and $\cn{z} = \cn{ac}$; see Fig.~\ref{fig:case4b5} where we have also assumed a coloring started w.l.o.g.\  with green $(\tw{a_1}, \cn{ab})$, red $(\tw{a_1}, \cn{ac})$, and blue $(\tw{a_2}, \cn{ab})$.

\begin{figure}[t]
	\centering
	\subfloat[\label{fig:case4b4}{}]{
	\includegraphics[page=4,width=0.19\textwidth]{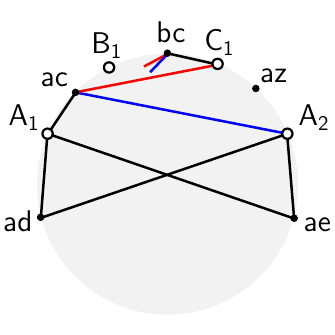}}
	\hfil
	\subfloat[\label{fig:case4b5}{}]{
	\includegraphics[page=5,width=0.19\textwidth]{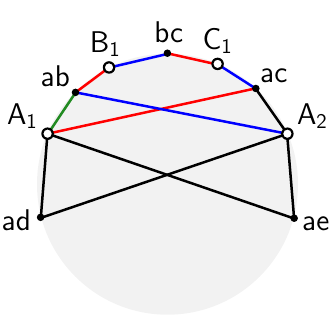}}
	\hfil
	\subfloat[\label{fig:case4b6}{}]{
	\includegraphics[page=6,width=0.19\textwidth]{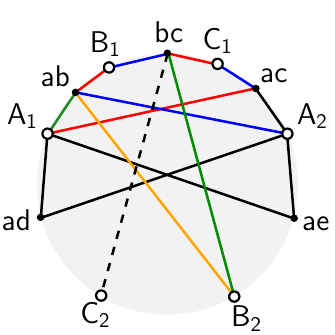}}
	\hfil
	\subfloat[\label{fig:case4b7}{}]{
	\includegraphics[page=7,width=0.19\textwidth]{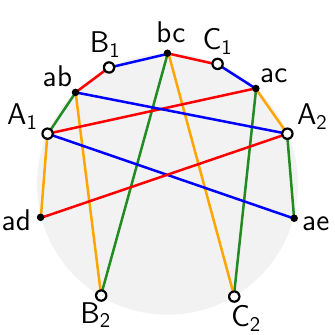}}
	\hfil
	\subfloat[\label{fig:case4b8}{}]{
	\includegraphics[page=8,width=0.19\textwidth]{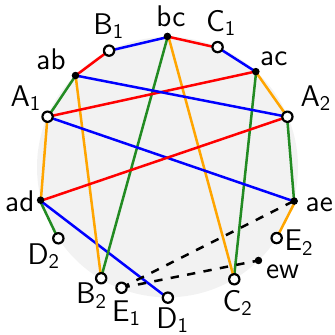}}
	\caption{Illustration of the case~$\cn{bc}=\cn{y} \in [\tw{a_1}, \tw{b_1}]$.
	%Black indicate unknown (not yet determined) edge color.
	}
\end{figure}

What is the placement of $\tw{b_2}$ and $\tw{c_2}$? Applying Lemma~\ref{lm:cross}.\ref{c:cross-1} for $\tw{b}$ and $\tw{a}$ and then for $\tw{c}$ and $\tw{a}$, we conclude that both $\tw{b_2}$ and $\tw{c_2}$ are in $[\cn{ae}, \cn{ad}]$. Edge $(\cn{ab}, \tw{b_2})$ is either orange or green; assume w.l.o.g.\ orange, and thus $(\cn{bc}, \tw{b_2})$ is green. If $\tw{b_2}$ appears before $\tw{c_2}$ in $[\cn{ae}, \cn{ad}]$, then the orange edge $(\cn{ab}, \tw{b_2})$ must cross the edge $(\cn{bc}, \tw{c_2})$; see Fig.~\ref{fig:case4b6}. Hence, $(\cn{bc}, \tw{c_2})$ needs a fifth color; contradiction. So, in the following we will assume that the order in $[\cn{ae}, \cn{ad}]$ is $(\dots \cn{ae} \dots \tw{c_2} \dots \tw{b_2} \dots \cn{ad} \dots)$. This order also fixes the colors of the following edges: $(\tw{C_2}, \cn{bc})$ is orange, $(\tw{c_2}, \cn{ac})$ is green, $(\tw{a_2}, \cn{ac})$ is orange, $(\tw{a_2}, \cn{ad})$ is red, $(\tw{a_1}, \cn{ae})$ is blue, $(\tw{a_1}, \cn{ad})$ is orange, $(\tw{a_2}, \cn{ae})$ is green; see Fig.~\ref{fig:case4b7}.

Since $\cn{bc} \in [\tw{b_1},\tw{c_1}]$, by applying Lemma~\ref{lm:noncross}.\ref{c:noncross-1} twice, it follows that there exists a $\tw{b}$'s and a $\tw{c}$'s connector in $[\tw{c_2}, \tw{b_2}]$. Thus, by Lemma~\ref{lm:alternate}, there is also a twin vertex in $[\tw{c_2}, \tw{b_2}]$. Up to renaming, we may assume that this twin is $\tw{d_1}$; observe that $\tw{d_1}$ is connected to $\cn{ad}$ via blue edge. By Lemma~\ref{lm:noncross}.\ref{c:noncross-3} and~\ref{c:noncross-4} applied for $\tw{d}$ and $\tw{a}$, we conclude that the second twin $\tw{d_2}$ cannot be in $[\cn{ad}, \tw{a_1}]$, because $\tw{b_2} \in [\tw{d_1}, \cn{ad}]$. Hence, by Lemma~\ref{lm:noncross}.\ref{c:noncross-1} and~\ref{c:noncross-2}, $\tw{d_2}$ is in $[\tw{a_2}, \cn{ad}]$. Since $\tw{d_2}$ is connected to $\cn{ad}$ via green edge, we can further constraint the placement of $\tw{d_2}$ in $[\tw{b_2}, \cn{ad}]$.

Since by Forbidden Pattern~\ref{case3} the order $(\ldots \tw{d_1} \cdot \tw{b_2} \cdot \tw{d_2} \ldots)$ is not possible, one of $\tw{e}$'s twins is in $[\tw{d_1}, \tw{d_2}]$, say w.l.o.g.\ $\tw{E_1}$. By Lemma~\ref{lm:noncross}.\ref{c:noncross-3} and~\ref{c:noncross-4} applied for $\tw{e}$ and $\tw{a}$, the second twin $\tw{e_2}$ cannot be in $[\tw{a_1}, \cn{ae}]$, because $\tw{b_2} \in [\cn{ae},\tw{E_2}]$. Hence, $\tw{e_2}$ is in $[\cn{ae},\tw{a_1}]$. Now, observe that $\cn{ae}$ must be connected to either $\tw{e_1}$ or $\tw{e_2}$ with an orange edge. This connection cannot be with $\tw{e_1}$, as othewrwise the orange edge $(\cn{ae},\tw{E_1})$ would cross $(\cn{bc},\tw{c_1})$, which is also orange; a contradiction. It follows that $\cn{ae}$ is connected to $\tw{e_2}$ with an orange edge, which further constraints the placement of $\tw{e_2}$ to $[\tw{a_2},\tw{c_2}]$; see Fig.~\ref{fig:case4b8}. We assume that $\tw{e_2}$ belongs to $[\cn{ae},\tw{c_2}]$; the case in which $\tw{e_2}$ belongs to  $[\tw{a_2},\cn{ae}]$ is similar.

By Lemma~\ref{lm:alternate}, there is a connector in $[\tw{e_2}, \tw{c_2}]$. This connector cannot be a connector of $\tw{d}$, as otherwise both of its connections with $\tw{d_1}$ and $\tw{d_2}$ must be red; a contradiction. Similarly, we can argue that the connector in $[\tw{e_2}, \tw{c_2}]$ can neither be a connector of $\tw{b}$ nor a connector of $\tw{c}$.  Thus, the connector in $[\tw{e_2}, \tw{c_2}]$ is a $\tw{e}$'s connector, call it $\cn{ew}$. Since $\tw{e_1}$ is connected to both $\cn{ae}$ and $\cn{ew}$, the two edges $(\tw{e_1}, \cn{ae})$ and $(\tw{e_1}, \cn{ew})$ need to cross a blue $(\tw{d_1}, \cn{ad})$, a green $(\tw{c_2}, \cn{ac})$, and an orange $(\tw{c_2}, \cn{bc})$ edge, which is impossible.\qed
\end{itemize}
\end{proof}

\begin{pattern}[$\boldmath \dots \tw{a_1} \cdot \tw{b_1} \dots \tw{b_2} \cdot \tw{a_2} \dots$]
\label{caseX}
It is impossible to have a non-crossing pair of adjacent twins.
\end{pattern}
\begin{proof}
Assume to the contrary that $\tw{a_1},\tw{a_2}$ and $\tw{b_1},\tw{b_2}$ form a non-crossing pair of adjacent twins. Neither $\tw{a_1},\tw{a_2}$ nor $\tw{b_1},\tw{b_2}$ are close twins, as otherwise they would form Forbidden Pattern~\ref{case4a}. Thus, by Lemma~\ref{lm:noncross}.\ref{c:noncross-1} and~\ref{c:noncross-2}, $\cn{ab} \in [\tw{a_1}, \tw{b_1}] \cup [\tw{b_2}, \tw{a_2}]$; say w.l.o.g.\  $\cn{ab} \in [\tw{b_2}, \tw{a_2}]$. By Lemma~\ref{lm:alternate}, there is a connector in $[\tw{a_1},\tw{b_1}]$, which might be adjacent to one of $\tw{a_1}$ or $\tw{b_1}$ or not. We consider each case separately. 

Assume first that the connector in $[\tw{a_1},\tw{b_1}]$ is adjacent neither to $\tw{a_1}$ nor to $\tw{b_1}$; w.l.o.g.\  assume it is $\cn{de}$. Since $\tw{a}$ and $\tw{b}$ are in 2-2 configuration, by Lemma~\ref{lm:22}, there exists a $\tw{a}$'s connector, say $\cn{ax}$, in $[\tw{b_1}, \tw{b_2}]$ and a $\tw{b}$'s connector, say $\cn{by}$, in $[\tw{a_2}, \tw{a_1}]$; see Fig.~\ref{fig:caseX1} where we have also assumed a coloring started w.l.o.g.\  with red $(\cn{ab}, \tw{a_1})$, orange $(\cn{ab}, \tw{a_2})$, blue $(\cn{ab}, \tw{b_1})$, and green $(\cn{ab}, \tw{b_2})$. Then, edges $(\tw{a_2}, \cn{ax})$ and $(\tw{b_1}, \cn{by})$ have to be green, while $(\tw{b_2}, \cn{by})$ and $(\tw{a_1}, \cn{ax})$ orange.

\begin{figure}[t]
	\centering
	\subfloat[\label{fig:caseX1}{}]{
	\includegraphics[page=1, width=0.19\textwidth]{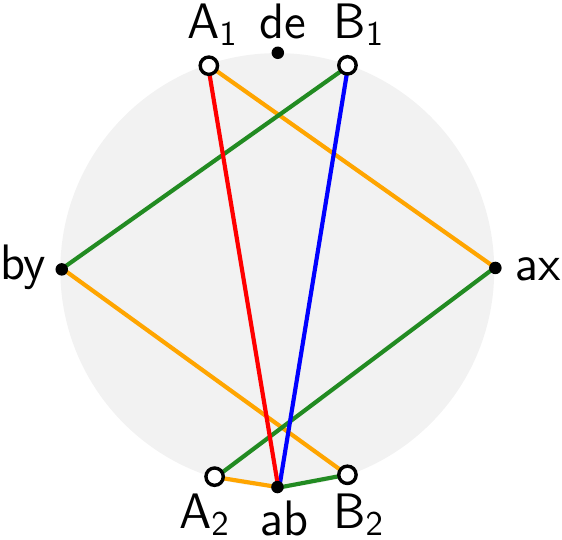}}
	\hfil
	\subfloat[\label{fig:caseX2}{}]{
	\includegraphics[page=2, width=0.19\textwidth]{caseX}}
	\caption{Illustrations for the proof of Forbidden Pattern~\ref{caseX}.
    %; black indicate unknown (not yet determined) edge color.
    }
\end{figure}

The two twins $\tw{x_1}$ and $\tw{x_2}$ adjacent to $\cn{ax}$ are in $[\tw{b_1}, \tw{b_2}]$, as they are connected to $\cn{ax}$ via red and blue edges. Similarly, twins $\tw{y_1}$ and $\tw{y_2}$  adjacent to $\cn{by}$ are in $[\tw{a_2}, \tw{a_1}]$. So, either $\cn{x}$ is $\cn{d}$ or $\cn{e}$, or $\cn{y}$ is $\cn{d}$ or $\cn{e}$. Further, we can conclude that the two connections from $\cn{de}$ to the interval $[\tw{b_1}, \tw{b_2}]$ go to the same twin pair, say $\tw{d_1},\tw{d_2}$, while the two connections to the interval $[\tw{a_2}, \tw{a_1}]$ go to twin pair $\tw{e_1},\tw{e_2}$. 

If both $\tw{c_1}$ and $\tw{c_2}$ were either in $[\tw{a_2}, \tw{a_1}]$ or in $[\tw{b_1}, \tw{b_2}]$, then $\tw{e_1}, \tw{e_2}$ would form pattern $(\ldots \tw{a_2} \cdot \tw{e_1} \cdot \tw{e_2} \cdot \tw{a_1}\ldots)$, which is forbidden by Forbidden Pattern~\ref{case4a}. Hence, we may assume w.l.o.g.\ that $\tw{c_1}$ is in $[\tw{a_2}, \tw{a_1}]$, while $\tw{c_2}$ in $[\tw{b_1}, \tw{b_2}]$. Now, observe that at most one of the edges incident to $\tw{c_1}$ might end in $[\tw{b_1}, \tw{b_2}]$, while its three remaining edges must end in $[\tw{a_2}, \tw{a_1}]$. A symmetric argument for $\tw{c_2}$ implies that at most one of its incident edges might end in the opposite interval $[\tw{a_2}, \tw{a_1}]$, while its three remaining edges must end in $[\tw{b_1}, \tw{b_2}]$. Since twins $\tw{c_1}$ and $\tw{c_2}$ share the same neighborhood, we have obtained a contradiction.

To complete the proof, we now consider the case in which the connector in $[\tw{a_1},\tw{b_1}]$ is adjacent to one of $\tw{a_1}$ or $\tw{b_1}$; assume w.l.o.g. that it is $\cn{ac}$ (recall that $\cn{ab} \in [\tw{b_2},\tw{a_2}]$). Again, $\tw{a}$ and $\tw{b}$ are in 2-2 configuration. Since connector $\cn{ab}$ is in $[\tw{B_2},\tw{B_1}]$, by Lemma~\ref{lm:22}, we may assume that $\tw{b}$ has one additional connector, say $\cn{bu}$, in $[\tw{B_2},\tw{B_1}]$, and two connectors, say $\cn{bv},\cn{bw}$, in $[\tw{B_1},\tw{B_2}]$. Since $\cn{ab} \in [\tw{B_2},\tw{A_2}]$, $\cn{ac} \in [\tw{A_1},\tw{B_2}]$, and since $\tw{a_1},\tw{a_2}$ and $\tw{b_1},\tw{b_2}$ form a non-crossing pair of adjacent twins, $\cn{bu} \in [\tw{A_2},\tw{A_1}]$. Fig.~\ref{fig:caseX2} illustrates the embedding with edge colors assigned as follows. W.l.o.g.\ $(\cn{ab}, \tw{a_1})$ is red, $(\cn{ab}, \tw{b_1})$ is blue, $(\tw{b_2}, \cn{bu})$ is orange, and $(\tw{b_1}, \cn{bu})$ is green. Then, $(\tw{a_2}, \cn{ac})$ is blue, $(\tw{b_2}, \cn{ab})$ is green, $(\tw{a_2}, \cn{ab})$ is orange. It follows that $(\tw{b_1}, \cn{bv})$ and $(\tw{b_2}, \cn{bw})$ are red, $(\tw{b_2}, \cn{bv})$ is blue, and $(\tw{b_1}, \cn{bw})$ is orange.

Next, we consider connector $\cn{bu}$ and, as in the previous case, we observe that both twin vertices $\tw{u_1}$ and $\tw{u_2}$ that are adjacent to $\cn{bu}$ must be on the interval $[\tw{a_2}, \tw{a_1}]$. If there were no other twin vertex on the interval $[\tw{a_2}, \tw{a_1}]$, then $\tw{u_1}$ and $\tw{u_2}$ would form pattern $(\ldots \tw{a_2} \cdot \tw{U_1} \cdot \tw{U_2} \cdot \tw{a_1}\ldots)$, which is forbidden by Forbidden Pattern~\ref{case4a}. Assume w.l.o.g.\ that twin vertex $\tw{v_1}$ is in $[\tw{a_2}, \tw{a_1}]$. We now claim that twin vertex $\tw{v_2}$ is also in $[\tw{b_1}, \tw{b_2}]$. To see this, first observe that, by Lemma~\ref{lm:alternate}, there exist at least one twin vertex in $[\tw{b_1}, \tw{b_2}]$, since both connectors $\cn{bv}$ and $\cn{bw}$ belong to this interval. Now, assume for a contradiction that twin vertex $\tw{v_2}$ is not in $[\tw{b_1}, \tw{b_2}]$. It follows that at most two twin vertices are in $[\tw{b_1}, \tw{b_2}]$, since $\tw{u_1},\tw{u_2},\tw{v_1},\tw{v_2} \in [\tw{a_2}, \tw{a_1}]$. If there is only one twin vertex in $[\tw{b_1}, \tw{b_2}]$, say $\tw{w_1}$, then pattern $(\ldots \tw{b_1} \cdot \tw{w_1} \cdot \tw{b_2}\ldots)$ is formed, which is forbidden by Forbidden Pattern~\ref{case3}. If there are two twin vertices, say $\tw{w_1}$ and $\tw{w_2}$, in $[\tw{b_1}, \tw{b_2}]$, then pattern $(\ldots \tw{b_1} \cdot \tw{w_1} \cdot \tw{w_2} \cdot \tw{b_2}\ldots)$ is formed, which is forbidden by Forbidden Pattern~\ref{case4a}. Hence, twin vertex $\tw{v_2}$ is in $[\tw{b_1}, \tw{b_2}]$, as claimed. Now, observe that out of the four edges incident to $\tw{v_1}$ at least three have to end in the interval $[\tw{a_2}, \tw{a_1}]$, where $\tw{v_1}$ resides. On the other hand, at most two edges incident to twin vertex $\tw{v_2}$ may end in the opposite interval $[\tw{a_2}, \tw{a_1}]$. Since twins $\tw{v_1}$ and $\tw{v_2}$ share the same neighbourhood, we have obtained a contradiction. This concludes the proof.\qed
\end{proof}

\begin{pattern}[$\boldmath \dots \tw{a_1} \cdot \tw{b_1}  \cdot \tw{c_1}\dots \tw{a_2}\cdot \tw{b_2} \cdot \tw{c_2} \dots$]
\label{ctriple}
It is impossible to have a \emph{crossing triple}, i.e., a triple of consecutive twins that pairwise cross.
\end{pattern}
\begin{proof}
Assume to the contrary that there exists a crossing triple, and let the order be $(\dots \tw{a_1}~\cn{x}~\tw{b_1}~\cn{y}~\tw{c_1} \dots \tw{a_2}~\cn{u}~\tw{b_2}~\cn{v}~\tw{c_2} \dots)$, where $\cn{x},\cn{y},\cn{u},\cn{v}$ are intermediate connectors. Observe that $\tw{a_1},\tw{a_2}$, $\tw{b_1},\tw{b_2}$, and $\tw{c_1},\tw{c_2}$ form three pairs of crossing twins. By Lemma~\ref{lm:cross}, the number of $\tw{b}$'s connectors on opposite intervals formed by a pair of crossing twins is the same. Thus, the number of $\tw{b}$'s connectors in $[\tw{a_1},\tw{c_1}]$ equals the number of $\tw{b}$'s connectors in $[\tw{a_2},\tw{c_2}]$, which implies that in the union of the two intervals there are in total either four, or two, or zero $\tw{b}$'s connectors (i.e., four, two, or zero out of $\cn{x},\cn{y},\cn{u},\cn{v}$ are $\tw{b}$'s connectors). We refer to the first and second case as \emph{non-zero crossing triple}, while to the third as \emph{zero crossing triple}. In Cases~(i) and~(ii) in the following we eliminate the case of non-zero crossing triples, and then assuming that there is no non-zero crossing triple we also eliminate the case of zero crossing triples. 
\begin{enumerate}[i.]
\item \emph{$\tw{b}$ has four connectors among $\cn{x},\cn{y},\cn{u},\cn{v}$}. By symmetry, we may assume $\cn{x} = \cn{ab}$. This implies that there is not a $\tw{a}$'s connector in $[\tw{a_2}, \tw{c_2}]$. By Lemma~\ref{lm:cross} for $\tw{a}$ and $\tw{c}$, however, it follows that there must exist a $\tw{a}$'s connector in $[\tw{a_2}, \tw{c_2}]$; a contradiction.

\item \emph{$\tw{b}$ has two connectors among $\cn{x},\cn{y},\cn{u},\cn{v}$}. 
By symmetry, we may assume that  the $\tw{b}$'s connectors are $\cn{x} \in [\tw{a_1},\tw{b_1}]$ and~$\cn{u} \in [\tw{a_2},\tw{b_2}]$. We will now prove by contradiction that $\cn{x} \notin \{\cn{ab},\cn{bc}\}$. Assume first that $\cn{x} = \cn{bc}$ and let w.l.o.g.\ the color of $(\tw{c_1},\cn{bc})$ be blue. Since $(\tw{c_1},\cn{bc})$ cannot be crossed by another blue edge, it follows that $(\tw{b_1}, \cn{y})$ exists and is blue. This, however, contradicts the fact that $\tw{B}$ has two connectors among $\cn{x},\cn{y},\cn{u},\cn{v}$. Assume now that $\cn{x} = \cn{ab}$. Since $\delta_{\tw{a}}(\tw{a_1}, \tw{b_1})=1$, it follows by Lemma~\ref{lm:cross}.\ref{c:cross-1} that $\delta_{\tw{a}}(\tw{a_2}, \tw{b_2})=1$. Hence, $\cn{u} \in [\tw{a_2}, \tw{b_2}]$ is a connector of $\tw{a}$. Since $\cn{ab} \in [\tw{a_1}, \tw{b_1}] $ and since by our initial assumption $\cn{u}$ is a connector of $\tw{b}$, we have again obtained a contradiction. It follows that either $\cn{x} = \cn{bd}$ or $\cn{x} = \cn{be}$ holds. By symmetry, either $\cn{u} = \cn{bd}$ or $\cn{u} = \cn{be}$ holds. By Lemma~\ref{lm:22}, there is a connector of $\tw{b}$ in each of $[\tw{c_1}, \tw{a_2}]$ and $[\tw{c_2}, \tw{a_2}]$. W.l.o.g.\ assume $\cn{ab}$ is in $[\tw{c_1},\tw{a_2}]$ and $\cn{bc}$ is in $[\tw{c_2},\tw{a_1}]$. It is easy now to see that the following edges pairwise cross: $(\tw{b_1},\cn{bc})$, $(\tw{c_1},\cn{bc})$, $(\tw{a_1},\cn{ab})$, $(\tw{b_1},\cn{ab})$, and $(\cn{u},\tw{b_1})$; a contradiction.

\item \emph{$\tw{b}$ has zero connectors among $\cn{x},\cn{y},\cn{u},\cn{v}$}. By~(i) and~(ii), it follows that no non-zero crossing triple exists. By Lemma~\ref{lm:cross}, two connectors of $\tw{b}$ exist in each of $[\tw{c_1},\tw{a_2}]$ and $[\tw{c_2},\tw{a_1}]$. Note that, $\cn{x}$ is not a connector of $\tw{c}$, as otherwise the four edges incident to $\tw{b_1}$ would cross $(\tw{c_1},\cn{x})$. By symmetry, $\cn{u}$ is not a connector of $\tw{c}$, and $\cn{y}$ and $\cn{v}$ are not connectors of $\tw{a}$. Also, $\cn{ac} \notin [\tw{a_1},\tw{c_1}] \cup [\tw{a_2},\tw{c_2}]$. 

Let $\delta(\tw{c_1},\tw{a_2})$ and $\delta(\tw{c_2},\tw{a_1})$ be the number of twins in $[\tw{c_1},\tw{a_2}]$ and $[\tw{c_2},\tw{a_1}]$, respectively. Clearly, $\delta(\tw{c_1},\tw{a_2})+\delta(\tw{c_2},\tw{a_1}) \le 4$ holds.  Since there exist two connectors of $\tw{b}$ in each of $[\tw{c_1},\tw{a_2}]$ and $[\tw{c_2},\tw{a_1}]$, there exist at least one twin in each of $[\tw{c_1},\tw{a_2}]$ and $[\tw{c_2},\tw{a_1}]$. Thus, $\delta(\tw{c_1},\tw{a_2}),\delta(\tw{c_2},\tw{a_1}) \ge 1$. Assume w.l.o.g.~that $\tw{d_1} \in [\tw{c_1},\tw{a_2}]$, and that $\tw{d_1}$ encountered first in $[\tw{c_1},\tw{a_2}]$.

The first twin encountered in $[\tw{c_2},\tw{a_1}]$ cannot be $\tw{d_2}$, as otherwise $\tw{b_1},\tw{c_1},\tw{d_1}$, and $\tw{b_2},\tw{c_2},\tw{d_2}$ would form a non-zero crossing triple containing connectors of $\tw{c}$. By symmetry, let $\tw{e_1}$ be the first twin in $[\tw{c_2},\tw{a_1}]$. 

We claim that $\delta(\tw{c_1},\tw{a_2}),\delta(\tw{c_2},\tw{a_1}) \leq 2$. For a contradiction, let $\delta(\tw{c_1},\tw{a_2})=3$ (the case $\delta(\tw{c_2},\tw{a_1})=3$ is symmetric). Then, $[\tw{c_1},\tw{a_2}]$ contains $\tw{d_1},\tw{d_2},\tw{e_2}$. If $\tw{d_2}$ precedes $\tw{e_2}$ in $[\tw{c_1},\tw{a_2}]$, then $\tw{e_1}, \tw{a_1},\tw{b_1}$, and $\tw{e_2},\tw{a_2},\tw{b_2}$ form a non-zero crossing triple containing connectors of $\tw{a}$. Otherwise, $\tw{d_2}$ follows $\tw{e_2}$ and thus $\tw{d_1}, \tw{e_2},\tw{d_2}$ form Forbidden Pattern~\ref{case3}. Hence, our claim holds.

Since $\tw{d_1} \in [\tw{c_1},\tw{a_2}]$,  $\tw{e_1} \in [\tw{c_2},\tw{a_1}]$ and $\delta(\tw{c_1},\tw{a_2}) \leq 2$, it follows that either $\tw{d_1},\tw{e_2} \in [\tw{c_1},\tw{a_2}]$ or $\tw{d_1},\tw{d_2} \in [\tw{c_1},\tw{a_2}]$ holds.  In the former case, $\tw{d}$ and $\tw{e}$ form Forbidden Pattern~\ref{caseX}. In the later case, $\tw{e_1},\tw{e_2} \in [\tw{c_2},\tw{a_1}]$, and the order is $(\tw{a_1} \cdot \tw{b_1} \cdot \tw{c_1} \cdot \tw{d_1} \cdot \tw{d_2} \cdot \tw{a_2} \cdot \tw{b_2} \cdot \tw{c_2} \cdot \tw{e_1} \cdot \tw{e_2} \cdot)$. Now, recall that $\cn{ac} \notin [\tw{a_1},\tw{c_1}] \cup [\tw{a_2},\tw{c_2}]$. By Lemma~\ref{lm:22}, it follows that $\cn{ac} \notin [\tw{d_1},\tw{d_2}]$ and $\cn{ac} \notin [\tw{e_1},\tw{e_2}]$. Hence, $\cn{ac}$ belongs to one of $[\tw{c_1},\tw{d_1}]$, $[\tw{d_2},\tw{a_2}]$, $[\tw{c_2},\tw{e_1}]$, $[\tw{e_2},\tw{a_1}]$. Assume the former; the remaining cases are similar. In this case,~$(\tw{c_1},\cn{ac})$ is crossed by the four edges incident to $\tw{b_1}$, which is not possible.\qed 
\end{enumerate}
\end{proof}

\begin{pattern}[$\boldmath \dots \tw{a_1} \cdot \tw{b_1} \dots \tw{a_2} \cdot \tw{b_2} \dots$]
\label{caseY}
It is impossible to have a crossing pair of adjacent twins.
\end{pattern}
\begin{proof}
Assume to the contrary that $\tw{a_1},\tw{a_2}$ and $\tw{b_1},\tw{b_2}$ is a crossing pair of adjacent twins. By Forbidden Patterns~\ref{case3} and~\ref{case4b}, there exist at least two twins in each of $[\tw{b_1},\tw{a_2}]$ and $[\tw{b_2},\tw{a_1}]$; call them $\tw{X},\tw{Y},\tw{U},\tw{V}$, and assume that the order is $(\dots \tw{x} \cdot \tw{a_1} \cdot \tw{b_1} \cdot \tw{y} \dots \tw{u} \cdot \tw{a_2} \cdot \tw{b_2} \cdot \tw{v} \dots)$. By Forbidden Pattern~\ref{case4b}, each of $\tw{x}, \tw{y}$ and $\tw{u}, \tw{v}$ are different twins. Let w.l.o.g.\ $\tw{x} = \tw{d_1}$, $\tw{y} = \tw{c_1}$. Since $\tw{u},\tw{v}$ are different, one of them, say w.l.o.g.~$\tw{u}$, is not a twin of $\tw{e}$. So, either $\tw{u}=\tw{c_2}$ or $\tw{u}=\tw{d_2}$ holds. If $\tw{u}=\tw{d_2}$, then the order is $(\dots \tw{d_1} \cdot \tw{a_1} \cdot \tw{b_1} \dots \tw{d_2} \cdot \tw{a_2} \cdot \tw{b_2} \dots)$, which implies that $\tw{d},\tw{a},\tw{b}$ form Forbidden Pattern~\ref{ctriple}. We conclude that $\tw{u}=\tw{c_2}$ holds; see Fig.~\ref{fig:caseYa}. 

Since the remaining twins are $\tw{d_2}$, $\tw{e_1}$, and $\tw{e_2}$, and since one of these is $\tw{v}$, there exist either zero, or one, or two twins in $[\tw{c_1},\tw{c_2}]$. One yields Forbidden Pattern~\ref{case3}, while two yields either Forbidden Pattern~\ref{case4a} or Forbidden Pattern~\ref{case4b}, depending on whether the two twins are same or not, respectively. Hence, we may assume that $\tw{c_1}$ and $\tw{c_2}$ are close twins.

Since $\tw{c_1}$ and $\tw{c_2}$ are close twins, twins $\tw{d_2}$, $\tw{e_1}$, and $\tw{e_2}$ are all in $[\tw{b_2}, \tw{d_1}]$. Hence, their relative order in $[\tw{b_2}, \tw{d_1}]$ is: \begin{inparaenum}[]
\item $(\tw{e_1} \cdot \tw{d_2} \cdot \tw{e_2})$, or
\item $(\tw{e_2} \cdot \tw{d_2} \cdot \tw{e_1})$, or
\item $(\tw{e_1} \cdot \tw{e_2} \cdot \tw{d_2})$, or
\item $(\tw{e_2} \cdot \tw{e_1} \cdot \tw{d_2})$, or
\item $(\tw{d_2} \cdot \tw{e_1} \cdot \tw{e_2})$, or
\item $(\tw{d_2} \cdot \tw{e_2} \cdot \tw{e_1})$.
\end{inparaenum}
The first two yield Forbidden Pattern~\ref{case3}. The next two yield Forbidden Pattern~\ref{case4a}. By symmetry of the last two cases we may assume that the order is $(\tw{a_1}~\cn{x}~\tw{b_1}~\cn{y}~\tw{c_1}~\cn{z}~\tw{c_2}~\cn{u}~\tw{a_2}~\cn{v}~\tw{b_2} \cdot \tw{e_1} \cdot \tw{e_2} \cdot \tw{d_2} \cdot \tw{d_1} \cdot)$, where $\cn{x},\cn{y},\cn{z},\cn{u},\cn{v}$ are intermediate connectors; see Fig.~\ref{fig:caseYb}.

\begin{figure}[t]
	\centering
	\subfloat[\label{fig:caseYa}{}]{
	\includegraphics[page=1, width=0.19\textwidth]{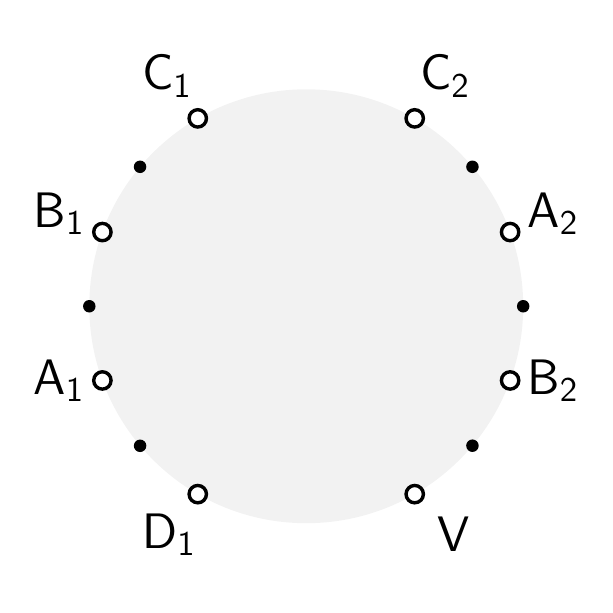}}
	\hfil
	\subfloat[\label{fig:caseYb}{}]{
	\includegraphics[page=2, width=0.19\textwidth]{caseY}}
	\caption{Illustration for the proof of Forbidden Pattern~\ref{caseY}}
\end{figure}

Since $\tw{c_1},\tw{c_2}$ and $\tw{d_1},\tw{d_2}$ are both close twins, by Corollary~\ref{lm:close}, it follows that connector $\cn{z}$, which is in $[\tw{c_1},\tw{c_2}]$, is not $\cn{cd}$. A symmetric argument on $\tw{c_1},\tw{c_2}$ and $\tw{e_1},\tw{e_2}$ implies that $\cn{z}$ is not $\cn{ce}$. By Lemma~\ref{lm:22}, $\cn{z}$ is either $\cn{ac}$ or $\cn{bc}$. By symmetry, we may assume $\cn{z}=\cn{ac}$. Since $\cn{ac} \in [\tw{c_1},\tw{c_2}]$, by Lemma~\ref{lm:noncross}.\ref{c:noncross-3} and~\ref{c:noncross-4} applied for $\tw{c}$ and $\tw{a}$, it follows that there exists a connector of $\tw{c}$ at each of the intervals $[\tw{c_1},\tw{c_2}]$, $[\tw{c_2},\tw{a_2}]$, $[\tw{a_2},\tw{a_1}]$ and $[\tw{a_1},\tw{c_1}]$. Thus, connector $\cn{u} \in [\tw{c_2},\tw{a_2}]$ is a connector of $\tw{c}$. Symmetrically, $\cn{y}\in[\tw{b_1},\tw{c_1}]$ is a connector of $\tw{c}$. By Lemma~\ref{lm:22} applied for $\tw{a}$, there are two $\tw{a}$'s connectors in $[\tw{a_1}, \tw{a_2}]$; thus, $\cn{x} \in [\tw{a_1},\tw{b_1}]$ is a connector of $\tw{a}$\footnote{Interval $[\tw{a_1},\tw{a_2}]$ is the union of $[\tw{a_1},\tw{b_1}]$, $[\tw{b_1},\tw{c_1}]$, $[\tw{c_1},\tw{c_2}]$, and $[\tw{c_2},\tw{a_2}]$. As in the last three intervals there exist connectors of $\tw{c}$ including $\cn{ac} \in [\tw{c_1},\tw{c_2}]$, it follows that the only interval, where the second connector of $\tw{a}$ can be, is $[\tw{a_1},\tw{b_1}]$.}. Symmetrically, $\cn{v}$ is a connector of $\tw{b}$. Since $\cn{v} \in [\tw{a_2},\tw{b_2}]$ is a connector of $\tw{b}$,  by Lemma~\ref{lm:cross}, it follows that $\cn{x} \in [\tw{a_1},\tw{b_1}]$ is also a connector of $\tw{b}$, which implies that $\cn{x} = \cn{ab}$ (recall that $\cn{x}$ is already shown to be a connector of $\tw{a}$). Since $\cn{x} \in [\tw{a_1},\tw{b_1}]$ is a connector of $\tw{a}$, again by Lemma~\ref{lm:cross}, it follows that $\cn{v} \in [\tw{a_2},\tw{b_2}]$ must be a connector of $\tw{a}$. Since we have already shown that $\cn{v}$ is a connector of $\tw{b}$, it follows that $\cn{v} = \cn{ab}$. This is a contradiction, as $\cn{ab} \in [\tw{a_1},\tw{b_1}]$.\qed
\end{proof}

\begin{pattern}[$\boldmath \dots \tw{a_1} \cdot \tw{b_1} \cdot \tw{c_1} \cdot \tw{d_1} \cdot \tw{a_2} \dots$]
\label{case5a}
Between any twin pair, it is impossible to have exactly three pairwise different twins.
\end{pattern}
\begin{proof}
Assume to the contrary that between $\tw{a_1}$ and $\tw{a_2}$ there exist exactly three pairwise different twins $\tw{b_1}, \tw{c_1}, \tw{d_1}$, that is, the order is $(\dots \tw{x} \cdot \tw{a_1} \cdot \tw{b_1} \cdot \tw{c_1} \cdot \tw{d_1} \cdot \tw{a_2} \cdot \tw{y} \dots)$, where $\tw{x}$ and $\tw{y}$ are the twins preceding $\tw{a_1}$ and following $\tw{a_2}$. If $\tw{x} = \tw{b_2}$, then $\tw{a}$ and $\tw{b}$ form Forbidden Pattern~\ref{case3};  if $\tw{x} = \tw{c_2}$, then $\tw{a}$ and $\tw{c}$ form Forbidden Pattern~\ref{case4b}; if $\tw{x} = \tw{d_2}$, then $\tw{d}$ and $\tw{a}$ form Forbidden Pattern~\ref{caseY}. Thus, $\tw{x} = \tw{e_1}$. By symmetry, $\tw{y} = \tw{e_2}$ holds. But then $\tw{a}$ and $\tw{e}$ form Forbidden Pattern~\ref{caseX}.\qed
\end{proof}

\begin{pattern}[$\boldmath \dots \tw{a_1} \cdot \tw{b_1} \cdot \tw{c_1} \cdot \tw{c_2} \cdot \tw{a_2} \dots$]
\label{case5b}
Between any twin pair, it is impossible to have exactly three twins, such that two of them form a pair.
\end{pattern}
\begin{proof}
Assume to the contrary that between $\tw{a_1}$ and $\tw{a_2}$ there exist exactly three twins, such that two of them form a pair. Let $\tw{b_1}, \tw{c_1}, \tw{c_2}$ be these twins. By Forbidden Pattern~\ref{case3}, $\tw{c_1}$ and $\tw{c_2}$ are consecutive, that is, the order is $(\tw{a_1} \cdot \tw{b_1} \cdot \tw{c_1} \cdot \tw{c_2} \cdot \tw{a_2} \cdot \tw{u} \cdot \tw{v} \cdot \tw{x} \cdot \tw{y} \cdot \tw{z} \cdot)$, where $\tw{u},\tw{v},\tw{x},\tw{y}$ and $\tw{z}$ are the (remaining) twins following $\tw{a_2}$. If $\tw{b_2} = \tw{z}$, then $\tw{a}$ and $\tw{b}$ form Forbidden Pattern~\ref{case3}; if $\tw{b_2} = \tw{y}$, then $\tw{a}$ and $\tw{b}$ form Forbidden Pattern~\ref{case4b}; if $\tw{b_2} = \tw{u}$, then $\tw{a}$ and $\tw{b}$ form Forbidden Pattern~\ref{caseY}. It follows that either $\tw{b_2} = \tw{v}$ or $\tw{b_2} = \tw{x}$ holds.

Assume first that $\tw{b_2} = \tw{v}$ holds. In this case, if twins $\tw{u}$ and $\tw{x}$ were of the same kind, then they would form Forbidden Pattern~\ref{case3}. Hence, $\tw{u}$ and $\tw{x}$ are different twins. Assume w.l.o.g.\ that $\tw{u} = \tw{d_1}$ and $\tw{x} = \tw{e_1}$. Thus, $\{\tw{y},\tw{z}\} = \{\tw{d_2},\tw{e_2}\}$. Clearly, if $\tw{d_2} = \tw{y}$ and $\tw{e_2} = \tw{z}$, then $\tw{d}$ and $\tw{e}$ form Forbidden Pattern~\ref{case3}. Thus, $\tw{e_2} = \tw{y}$ and $\tw{d_2} = \tw{z}$ holds. So, the order is $(\tw{a_1} \cdot \tw{b_1} \cdot \tw{c_1} \cdot \tw{c_2} \cdot \tw{a_2} \cdot \tw{d_1} \cdot \tw{b_2} \cdot \tw{e_1} \cdot \tw{e_2} \cdot \tw{d_2} \cdot)$. Now, it is not difficult to see that $\tw{a}$ and $\tw{d}$ form Forbidden Pattern~\ref{caseX}.

To complete the proof, assume now that $\tw{b_2} = \tw{x}$. If twins $\tw{u}$ and $\tw{v}$ were different, say w.l.o.g.\ that $\tw{u} = \tw{d_1}$ and $\tw{v} = \tw{e_1}$, then $\tw{E_1}$ and $\tw{E_2}$ would form Forbidden Pattern~\ref{case3} or~\ref{case4b}, as one of $\tw{y}$ and $\tw{z}$ must be $\tw{e_2}$. Hence, $\tw{u}$ and $\tw{v}$ are same twins. Symmetrically, $\tw{y}$ and $\tw{z}$ are also same twins. Assume w.l.o.g.\ that $\tw{d_1} = \tw{u}$, $\tw{d_2} = \tw{v}$, $\tw{e_1} = \tw{y}$, and $\tw{e_2} = \tw{z}$, that is, the order is $(\tw{a_1} \cdot \tw{b_1} \cdot \tw{c_1} \cdot \tw{c_2} \cdot \tw{a_2} \cdot \tw{d_1} \cdot \tw{d_2} \cdot \tw{b_2} \cdot \tw{e_1} \cdot \tw{e_2} \cdot)$. Note that none of our forbidden patterns is violated.
%Surprisingly, this order does not violate any of our forbidden patterns, which implies that we inevitably have to continue our case analysis.

By Lemma~\ref{lm:22}, the connector between $\tw{d_1}$ and $\tw{d_2}$ is a connector of $\tw{D}$. Since $\tw{e}$ and $\tw{c}$ are close twins, it follows by Corollary~\ref{lm:close}, that the connector between $\tw{d_1}$ and $\tw{d_2}$ is neither $\cn{cd}$ nor $\cn{de}$. Hence, it is one of $\cn{bd}$ and $\cn{ad}$; by symmetry, assume it is $\cn{bd}$. Since $\tw{a}$ and $\tw{b}$ form a crossing twin pair and since $\cn{bd} \in [\tw{a_2},\tw{b_2}]$, by Lemma~\ref{lm:cross}.\ref{c:cross-1} it follows that $\delta_{\tw{b}}(\tw{a_1}, \tw{b_1}) = 1$, which implies that the connector in $[\tw{a_1},\tw{b_2}]$ is a connector of $\tw{b}$; see Fig.~\ref{fig:case5b1}, where we have assumed w.l.o.g.\ that $(\tw{b_1}, \cn{bd})$ is green, $(\tw{b_2}, \cn{bd})$ is red, $(\tw{d_2}, \cn{bd})$ is blue, and $(\tw{d_1}, \cn{bd})$ is orange.

\begin{figure}[t]
	\centering
	\subfloat[\label{fig:case5b1}{}]{
	\includegraphics[page=1, width=0.19\textwidth]{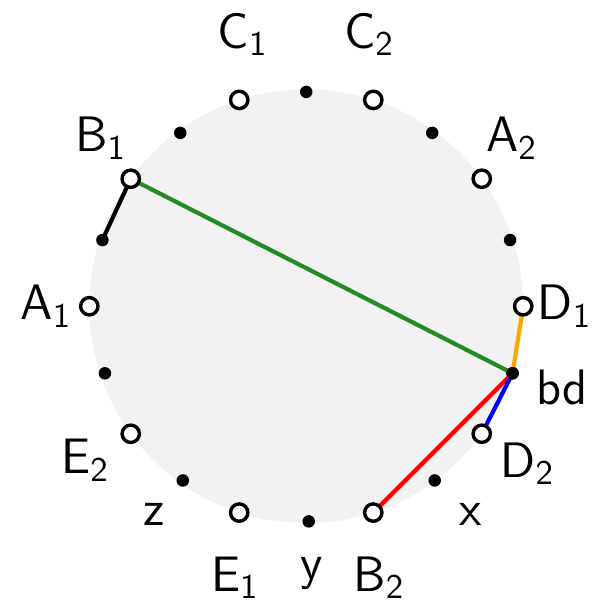}}
	\hfil
	\subfloat[\label{fig:case5b2}{}]{
	\includegraphics[page=2, width=0.19\textwidth]{case5b}}
	\hfil
	\subfloat[\label{fig:case5b3}{}]{
	\includegraphics[page=3, width=0.19\textwidth]{case5b}}
	\hfil
	\subfloat[\label{fig:case5b4}{}]{
	\includegraphics[page=4, width=0.19\textwidth]{case5b}}
	\hfil
	\subfloat[\label{fig:case5b5}{}]{
	\includegraphics[page=5, width=0.19\textwidth]{case5b}}
	\caption{
	Illustration for the proof of Forbidden Pattern~\ref{case5b}
	%; black indicate unknown (not yet determined) edge color.
}
\end{figure}

Let $\cn{x},\cn{y},\cn{z}$ be connectors in $[\tw{d_2}, \tw{b_2}]$, $[\tw{b_2},\tw{e_1}]$ and $[\tw{e_1},\tw{e_2}]$. Since $(\tw{b_2},\cn{bd})$ cannot be crossed by another red edge, it follows that $(\tw{d_2}, \cn{x})$ exists, and that it is red. Hence, $\cn{x}$ is a connector of $\tw{D}$, and it can be easily seen that $(\tw{d_1}, \cn{x})$ is blue; Fig.~\ref{fig:case5b2} for an illustration. Since $\cn{bd} \in [\tw{d_1},\tw{d_2}]$, connector $\cn{x}$ is either $\cn{ad}$, or $\cn{cd}$, or $\cn{de}$. If $\cn{x}$ were $\cn{ad}$, then the edge $(\tw{a_2}, \cn{x})$ must be inevitably orange, which implies that edge $(\tw{a_1}, \cn{x})$ is green. However, this color makes  impossible to route a green edge for $\tw{d_2}$, since the connector of $\tw{b}$ in $[\tw{a_1},\tw{b_1}]$ is not $\cn{bd}$; see Fig.~\ref{fig:case5b3}. If $\cn{x}$ were $\cn{cd}$, then both edges $(\tw{c_1}, \cn{x})$ and $(\tw{c_2}, \cn{x})$ would have to cross the green edge $(\tw{b_1}, \cn{db})$. However, this is a contradiction again, since one of them must indeed be green and orange; see see Fig.~\ref{fig:case5b4}. Therefore, $\cn{x} = \cn{de}$. This implies that edge $(\tw{b_2}, \cn{y})$ exists and it is of the same color as $(\tw{e_1}, \cn{de})$. Similarly, edge $(\tw{b_2}, \cn{z})$ exists and it is of the same color as $(\tw{e_2}, \cn{de})$; see Fig.~\ref{fig:case5b5}. Hence, both $\cn{y}$ and $\cn{z}$ are connectors of $\tw{B}$. So, including the connector of $\tw{b}$ in the interval $[\tw{a_1},\tw{a_2}]$ there exist in total three connectors of $\tw{B}$ on the interval $[\tw{b_2},\tw{b_1}]$, contradicting Lemma~\ref{lm:22}.\qed
\end{proof}

\begin{theorem}
\label{thm:folkman}
The dispersable book thickness of the Folkman graph is five.
\end{theorem}
\begin{proof}
To prove this theorem, it remains to conclude that the dispersable book thickness of the Folkman graph is not four. Let $d(\tw{a_1}, \tw{a_2})$ be the number of twin vertices in $[\tw{a_1}, \tw{a_2}]$, and let $d(\tw{a}) = \min \left( d(\tw{a_1}, \tw{a_2}), d(\tw{a_2}, \tw{a_1}) \right)$. By Forbidden Pattern~\ref{case3}, $d(\tw{a}) \neq 3$; by Forbidden Patterns~\ref{case4a} and~\ref{case4b}, $d(\tw{a}) \neq 4$; by Forbidden Patterns~\ref{case5a} and~\ref{case5b}, $d(\tw{a}) \neq 5$. Therefore, either $d(\tw{a}) = 2$ or $d(\tw{a}) = 6$, that is, two twins are either close or are opposite in a dispersable order.

Assume now that there is a pair of twins, say $\tw{a_1}, \tw{a_2}$, that are opposite, and let $\tw{X}, \tw{Y}, \tw{Z}, \tw{W}$ be the twins in $[\tw{a_1}, \tw{a_2}]$. If twin $\tw{X}$ and its counterpart were also opposite, then $\tw{A}$ and $\tw{X}$ would form Forbidden Pattern~\ref{caseX}. Hence, $\tw{X}, \tw{Y}$ are close. Symmetrically, $\tw{Z}, \tw{W}$ are also close. Hence, at most one pair of twins are opposite.

Fig.~\ref{fig:remain} illustrates the remaining two cases, in which either no or one pair of twins are opposite. In the former case, by Lemma~\ref{lm:22} there is a $\tw{A}$'s connector, say w.l.o.g.~$\cn{ab}$, in $[\tw{a_1}, \tw{a_2}]$. Then, by Corollary~\ref{lm:close}, twins $\tw{b_1}, \tw{b_2}$ must be far; a contradiction.

\begin{figure}[t]
	\centering
	\subfloat[\label{fig:remain1}{}]{
	\includegraphics[page=1, width=0.19\textwidth]{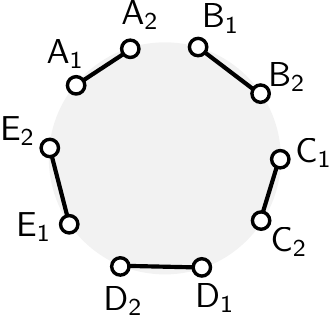}}
	\hfil
	\subfloat[\label{fig:remain2}{}]{
	\includegraphics[page=2, width=0.19\textwidth]{thmFolkman}}
	\caption{Illustration for the proof of Theorem~\ref{thm:folkman}}
	\label{fig:remain}
\end{figure}

To complete the proof, we find a contradiction for the case, in which there exists one pair of opposite twins, say w.l.o.g. $\tw{c_1}, \tw{c_2}$; see Fig.~\ref{fig:remain2}. By Corollary~\ref{lm:close}, the connector between each of the four close pairs can only be a connector of $\tw{C}$. Hence, the order is $(\tw{c_2}~\cn{x}~\tw{a_1}~\cn{ac}~\tw{a_2}~\cn{y}~\tw{b_1}~\cn{bc}~\tw{b_2}~\cn{z}~\tw{c_1}~\cn{u}~\tw{d_1}~\cn{dc}~\tw{d_2}~\cn{v}~\tw{e_1}~\cn{ec}~\tw{e_2}~\cn{w})$, where $\cn{x},\cn{y},\cn{z},\cn{u},\cn{v},\cn{w}$ are the remaining connectors in the order.

By Lemma~\ref{lm:noncross} applied for $\tw{a}$ and $\tw{c}$, twins $\tw{c_1}$ and $\tw{c_2}$ are separated by $\tw{a}$'s connectors; that is, there exist $\tw{a}$'s connectors on both intervals $[\tw{a_2}, \tw{c_1}]$ and $[\tw{c_2},\tw{a_1}]$. Thus, $\cn{x}$ is a connector of $\tw{a}$. Similarly, we conclude that $\cn{z}$ is $\tw{b}$'s connector. Next observe that $\cn{z} \neq \cn{ab}$, otherwise five edges, $(\cn{ab}, \tw{a_1})$, $(\cn{ab}, \tw{a_2})$, $(\cn{ab},\tw{b_1})$, $(\cn{bc},\tw{c_2})$, $(\cn{bc},\tw{c_1})$, would pairwise cross. Hence, $\cn{y} = \cn{ab}$. Arguing symmetrically, we find that $\cn{v}=\cn{de}$, $\cn{u}$ is $\tw{d}$'s connector, and $\cn{w}$ is $\tw{e}$'s connector. That means that connector $\cn{bd}$ is either $\cn{z}$ or $\cn{u}$. Both cases are impossible, as edges $(\tw{d_1},\cn{bd})$ or $(\tw{b_2},\cn{bd})$ would cross four $\tw{c_1}$'s edges; a contradiction.\qed
\end{proof}

\begin{corollary}
The Folkman graph is not dispersable.
\end{corollary}

\section{The Dispersable Book Thickness of the Gray Graph}
\label{sec:3-regular}
In this section, we study the book thickness of the Gray graph~\cite{BOUWER197232}, which can be constructed in two steps starting from three copies of $K_{3,3}$ as follows. First, we subdivide every edge as we did also with the Folkman graph (see Fig.~\ref{fig:gray_2}). Then, for each newly introduced vertex $u$ in the first copy, with $v$ and $w$ being its counterparts in the other two copies, we add a new vertex connected to $u$, $v$ and $w$  (see Fig.~\ref{fig:gray_3}). The resulting graph is the Gray graph, which is clearly $3$-regular and bipartite.

\begin{figure}[t]
	\centering
	\subfloat[\label{fig:gray_1}{}]
	{\includegraphics[page=1,width=0.3\textwidth]{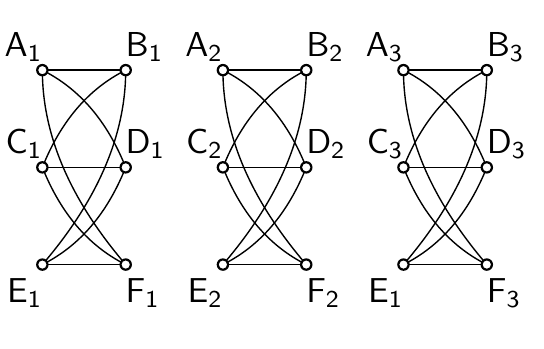}}
	\hfil 
	\subfloat[\label{fig:gray_2}{}]
	{\includegraphics[page=2,width=0.3\textwidth]{gray}}
	\hfil
	\subfloat[\label{fig:gray_3}{}]
	{\includegraphics[page=3,width=0.3\textwidth]{gray}}
	\caption{Construction steps for the Gray graph~\cite{FOLKMAN1967215}:
	(a)~the vertices of the three copies of $K_{3,3}$ are denoted by 
	$\tw{a}_i$, $\tw{b}_i$, $\tw{c}_i$, $\tw{d}_i$, $\tw{e}_i$, $\tw{f}_i$; $i=1,2,3$. 
	(b)~A vertex introduced in the first step between $\tw{X}_i$ and $\tw{Y}_i$ is denoted by $\cn{xy_i}$.
	(c)~A vertex introduced in the second step connecting $\cn{xy_1}$, $\cn{xy_2}$ and $\cn{xy_3}$ is denoted by $\tw{xy}$,
	 where $\tw{X},\tw{Y} \in \{\tw{a},\tw{b},\tw{c},\tw{d},\tw{f},\tw{e}\}$.}
	\label{fig:gray}
\end{figure}

Our computer-aided proof is based on appropriately adjusting a relatively recent formulation of the (ordinary) book embedding problem as a SAT instance~\cite{BKZ15}. In this formulation, there exist three different types of variables, denoted by $\sigma$, $\phi$ and $\chi$, with the following meanings: 
\begin{inparaenum}[(i)]
\item for a pair of vertices $u$ and $v$, variable $\sigma(u,v)$ is $\texttt{true}$, if and only if $u$ is to the left of $v$ along
the spine, 
\item for an edge $e$ and a page $i$, variable $\phi_i(e)$ is $\texttt{true}$, if and only if edge $e$ is assigned to page $i$ of the book, and 
\item for a pair of edges $e$ and $e'$, variable $\chi(e,e')$ is $\texttt{true}$, if and only if $e$ and $e'$ are assigned to the same page.
\end{inparaenum}  
Hence, there exist in total $O(n^2+m^2+pm)$ variables, where $n$ denotes the number of vertices of the graph, $m$ its number of edges, and $p$ the number of available pages. An additional $O(n^3+m^2)$ clauses ensure that the underlying order is indeed linear, and that no two edges of the same page cross; for details refer to~\cite{BKZ15}.

For the dispersable case, we must additionally guarantee that no two edges with a common endvertex are assigned to the same page. This requirement, however, can be easily encoded by the following clauses:
\begin{eqnarray*}
\neg \chi(e,e'), \quad \forall~e,e' \mbox{ with a common endvertex}
\end{eqnarray*}
Observe that there is no need to introduce new variables, and that the total number of constraints is not asymptotically affected. Using this adjustment, we proved that the dispersable book thickness of the Gray graph cannot be three, and that it admits a dispersable book embedding with four pages; see Fig.~\ref{fig:gray-dispersable} in Appendix~\ref{app:gray}. We summarize these findings in the following theorem.

\begin{theorem}
\label{thm:gray}
The dispersable book thickness of the Gray graph is four.
\end{theorem} 

\begin{corollary}
The Gray graph is not dispersable.
\end{corollary}

\section{$\boldmath 3$-connected $\boldmath 3$-regular Bipartite Planar Graphs}
\label{sec:3-regular-planar}
% !TEX root = bookmatching.tex
In the previous section, we demonstrated that the Gray graph, which is $3$-connected, $3$-regular and bipartite, is not dispersable. This graph, however, is not planar, as it contains $K_{3,3}$ as minor. In the following, we prove that when adding planarity to the requirements, every such graph is dispersable. For the sake of simplicity, we refer to a $3$-connected $3$-regular bipartite planar graph as \emph{Barnette graph} for short (due to the well-known Barnette's Conjecture~\cite{Bar69} which states that every such graph is Hamiltonian).

\begin{lemma}\label{lem:barnetteprops}
Let $G=(V, E)$ be an embedded Barnette graph and let $G^*=(V^*, E^*)$ be its dual. Then, there exists a $3$-edge coloring $\Ered \cupdot \Egreen \cupdot \Eblue = E$ for $G$, and a $3$-vertex coloring $\Dred \cupdot \Dgreen \cupdot \Dblue = V^*$ for $G^*$ which satisfy the following three properties: 
\begin{enumerate}[i.]
\item \label{prp:1} Every facial cycle of $G$ is \emph{bichromatic}, i.e., the edges on a facial cycle of $G$ alternate between two colors.
\item \label{prp:2} Every face of $G$ is colored differently from its bounding edges.
\item \label{prp:3} The edges of $G^*$ that connect vertices of $\Dgreen$ to vertices of $\Dblue$ are in one-to-one correspondence with the edges of $\Ered$, and induce a connected subgraph.
\end{enumerate}
\end{lemma}
\begin{proof}
Since $G$ is $3$-regular and bipartite, $G^*$ is maximal planar and every vertex has even degree. By the $3$-color theorem, $G^*$ is $3$-vertex colorable~\cite{steinberg93}. Also, since $G^*$ is maximal planar and its chromatic number is~$3$, $G^*$ is \emph{uniquely $3$-colorable}~\cite{cg69}, i.e., it has a unique $3$-vertex coloring up to permutation of the colors. Let $\Dred, \Dgreen$ and $\Dblue$ be this $3$-vertex coloring of $G^*$. 

We first show Property~(\ref{prp:2}).  Every edge $e$ of $G$ bounds two faces that are colored differently in $G^*$. Hence, we can assign to $e$ the third unused color. Since every vertex $v$ of $G$ is incident to three faces (which pairwise do not share the same color in $G^*$), no two edges incident to $v$ have the same color. Thus, the result is a proper $3$-edge coloring $\Ered, \Egreen, \Eblue$ of $G$.  

Now Property~(\ref{prp:1}) follows easily: On every facial cycle $p$ of $G$, two edges sharing an endpoint have distinct colors. By Property~(\ref{prp:2}), their colors are different from the color of $p$ in $G^*$. Thus, every face of $G$ is bichromatic.  

Next we show Property~(\ref{prp:3}). By Property~(\ref{prp:2}), any edge of $G^*$ that corresponds to an edge of $\Ered$ has one endpoint in $\Dgreen$ and one in $\Dblue$. Conversely, by construction every edge of $G^*$ in the induced subgraph of $\Dgreen \cup \Dblue$ corresponds to an edge in $\Ered$ of $G$. Hence, the edges of $G^*$ that connect vertices of $\Dgreen$ to vertices of $\Dblue$ are indeed in one-to-one correspondence with the edges of $\Ered$. Property~(\ref{prp:3}) follows by a result of Chartrand and Geller~\cite{cg69}, who showed that for any $k$-vertex coloring of a uniquely $k$-colorable graph, the subgraph induced by any two of the $k$ colors is connected.\qed
\end{proof}

In the following, we show that it is always possible to determine a dispersable order for a Barnette graph $G$ such that the coloring of Lemma~\ref{lem:barnetteprops} for $G$ is a feasible page assignment. In particular, the green edges will always connect vertices that are consecutive in the dispersable order, which implies that they can be ``merged'' with either the red or the blue edges, yielding thus a (non-dispersable) two-page book embedding of $G$. Our construction is based on determining a so-called \emph{subhamiltonian cycle} $C$ for $G$, that is, a cyclic ordering of the vertices such that when adding any missing edges between consecutive vertices the resulting graph remains planar (thus, $C$ becomes a Hamiltonian cycle). This subhamiltonian cycle partitions the edges of $G$, such that the green and the red edges are either inside or on $C$, while the blue edges are in the exterior of $C$. With these properties in mind, we now state the main result of this section.

\begin{theorem}\label{thm:barnette}
Let $G = (V,E)$ be an embedded Barnette graph. Then, there exists a $3$-edge coloring $\Ered \cupdot \Egreen \cupdot \Eblue = E$ of $G$ and a subhamiltonian cycle $C$ satisfying the following properties:
\begin{enumerate}[i.]
\item \label{prp:4} every $e \in \Ered$ is in the interior of $C$ or on $C$, 
\item \label{prp:5} every $e \in \Eblue$ is in the exterior of $C$ or on $C$,
\item \label{prp:6} every $e \in \Egreen$ is on $C$.
\end{enumerate}
\end{theorem}
\begin{proof}
In the proof, we assume that $\Ered \cupdot \Egreen \cupdot \Eblue$ is a $3$-edge coloring of $G$, and that $\Dred \cupdot \Dgreen \cupdot \Dblue$ is a $3$-vertex coloring of the dual $G^*=(V^*,E^*)$ of $G$ satisfying the Properties~\ref{prp:1}-\ref{prp:3} of Lemma~\ref{lem:barnetteprops}. By Lemma~\ref{lem:barnetteprops}.\ref{prp:3}, the subgraph $G_{bg}^*$ of $G^*$ induced by $\Dgreen \cup \Dblue$ is connected. Hence, we can construct a spanning tree $T^*$ of $G_{bg}^*$. This tree and the one-to-one correspondence between the edges of $\Ered$ and the edges of $G_{bg}^*$ yield a partition of $\Ered$ into two sets $\EredT$ and $\EredN$, such that $\EredT \cupdot \EredN = \Ered$, as follows. An edge $e \in \Ered$ belongs to $\EredT$, if the edge of $G^*$ corresponding to $e$ belongs to $\mathcal{T}^*$. Otherwise, $e$ belongs to $\EredN$. We also assume $\mathcal{T}^*$ to be rooted at a leaf  $\rho$, such w.l.o.g.\ $\rho \in \Dblue$.    

The proof is given by a recursive geometric construction of the subhamiltonian cycle $C$. Consider an arbitrary edge $(u,v) \in \EredT$ of $G$, and let $p$ and $q$ be the faces to its left and the right side, respectively, as we move along $(u,v)$ from $u$ to $v$. Then, $(p,q)$ is an edge of $\mathcal{T}^*$. Since $\mathcal{T}^*$ is a tree, the removal of $(p,q)$ results in two trees $\mathcal{T}^*_p$ and $\mathcal{T}^*_q$. W.l.o.g.\ we assume that $\rho$ belongs to $\mathcal{T}^*_p$. For the recursive step of our algorithm, we assume that we have already computed a \emph{simple} and \emph{plane} cycle $C_p$ for the subgraph $G_p=(V_p,E_p)$ of $G$ induced by the vertices of the faces of $G$ in $\mathcal{T}^*_p$, which satisfies the following additional invariants: 

\begin{enumerate}[{I.}1]
\item \label{i:1} edge $(u,v)$ is on $C_p$,
\item \label{i:2} every edge $e \in \EredT  \cap E_p$ is in the interior of $C_p$ or on $C_p$, 
\item \label{i:3} every edge $e \in \Eblue  \cap E_p$ is in the exterior of $C_p$ or on $C_p$, 
\item \label{i:4} every edge $e \in \Egreen \cap E_p$ is on $C_p$, and
\item \label{i:5} every edge $e \in \EredN$ that bounds two faces $h,h'$, with $h \in \mathcal{T}^*_p$ and $h' \notin \mathcal{T}^*_p$, is such that: 
\begin{enumerate}[i.]
\item \label{i:5.i}  if $h \in \Dblue$, then both endpoints of $e$ are on $C_p$,
\item \label{i:5.ii} if $h \in \Dgreen$, then none of the endpoints of $e$ is on $C_p$.
\end{enumerate}
\end{enumerate} 

\begin{figure}[t]
	\centering
	\subfloat[\label{fig:blue}{$q \in \Dblue$}]{
	\includegraphics[page=2,scale=0.5]{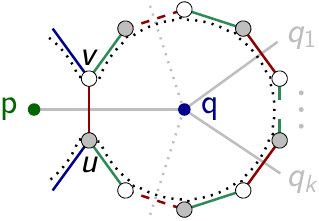}}
	\hfil
	\subfloat[\label{fig:green}{$q \in \Dgreen$}]{
	\includegraphics[page=3,scale=0.5]{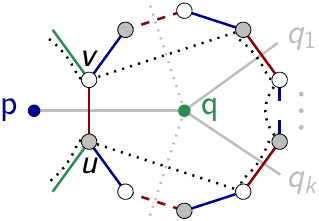}}
	\caption{Illustration for Theorem~\ref{thm:barnette}.
	The solid (dotted) gray edges belong to $\mathcal{T}^*$ ($G_{bg}^* \setminus \mathcal{T}^*$).	
	The solid (dashed) red edges belong to $\EredT$ ($\EredN$). 
	Cycle $C_q$ is drawn dotted black.
	}
	\label{fig:barnette}
\end{figure}

% Note that by Invariant I.\ref{i:5} we essentially require that $C_p$ contains all red edges of $\EredN$ that are on the facial cycle of every blue face. 
Let $G_q=(V_q,E_q)$ be the subgraph of $G$ induced by the vertices of the faces of $G$ in $\mathcal{T}^*_q$. Let also $q_1,\ldots,q_k$, with $k\geq 0$, be the children of $q$ in $\mathcal{T}^*$ (if any). We proceed by considering two cases; $q \in \Dblue$ and $q \in \Dgreen$; see Figs.~\ref{fig:blue} and~\ref{fig:green}, respectively. Note that by Lemma~\ref{lem:barnetteprops}.\ref{prp:1} and~\ref{lem:barnetteprops}.\ref{prp:2} in the former case, the edges of $q$ alternate between red and green, while in the latter case between red and blue.

Assume first that $q \in \Dblue$. We remove from $C_p$ the edge $(u,v)$, which exists by Invariant I.\ref{i:1}. This results in a path from $u$ to $v$. The cycle $C_q$ that is obtained by this path and the path from $u$ to $v$ in face $q$ is a cycle for $\mathcal{T}^*_p \cup \{q\}$, which satisfies Invariants I.\ref{i:1}-I.\ref{i:5} as we discuss in the following. 

Since $q$ is a face, it has no chords. Hence, the only chord we added to $C_q$ is $(u,v)$ which belongs to $\EredT$. This implies that Invariants I.\ref{i:2}-I.\ref{i:4} are satisfied. Since $q \in \Dblue$, in order to guarantee Invariant~\ref{i:5.i}, we have to ensure that all edges of $q$ that belong to $\EredN$ belong to $C_q$. However, this trivially holds, since every edge of $q$ is on $C_q$ except for edge $(u,v)$, which however belongs to $\EredT$ (and thus not to $\EredN$). To guarantee Invariant I.\ref{i:1}, we have to ensure that the edge of $G$ shared between $q$ and each child $q_i$ of $q$ in $\mathcal{T}^*$ is on $C_q$ ($i=1,\ldots,k$). Since the only edge of $q$ that is not on $C_q$ is the edge $(u,v)$, which is not shared by any child of $q$, it follows that Invariant I.\ref{i:1} is also maintained. It remains to also prove that $C_q$ is simple and plane. The latter property is straight-forward. For the former property, assume for a contradiction that cycle $C_q$ is not simple. Since $C_p$ is simple, one of the newly introduced vertices, call it $w$, is in $C_p$. Since $q \in \Dblue$, $w$ is incident to a red edge, say $(w,z)$, in $q$. Let $q' \in \Dgreen$ be the face of $G$ that is incident to $(w,z)$ and different from $q$. Since $w$ belongs to $C_p$, face $q'$ belongs to $\mathcal{T}^*_p$. Hence, $(w,z) \in \EredN$. Since $q' \in \Dgreen$, none of the endpoints of $(w,z)$ is on $C_p$ due to Invariant I.\ref{i:5.ii}; a contradiction. Hence, $C_q$ is indeed simple. 

Assume now that $q \in \Dgreen$. If $q$ is a leaf in $\mathcal{T}^*$ (i.e., the only edge incident to $q$ that belongs to $\EredT$ is edge $(u,v)$), then $C_p$ is a (simple and plane) cycle also for $\mathcal{T}^*_p \cup \{q\}$, which trivially satisfies Invariants I.\ref{i:1}-I.\ref{i:5}. So, we assume w.l.o.g.\ that $q$ is not a leaf in $\mathcal{T}^*$. It follows that there exist edges of $q$, different from $(u,v)$, that belong to $\EredT$. Denote by $w_1,\ldots,w_\ell$ the endvertices of these edges as they appear in a clockwise traversal of $q$ starting from $u$. We remove from $C_p$ the edge $(u,v)$, which exists by Invariant I.\ref{i:1}. This results in a path from $u$ to $v$. The cycle $C_q$ that is obtained by this path and the path $u \rightarrow w_1 \rightarrow \ldots \rightarrow w_\ell \rightarrow v$ is a cycle for $\mathcal{T}^*_p \cup \{q\}$, which satisfies Invariants I.\ref{i:1}-I.\ref{i:5} as we prove in the following. 

Since $C_q$ passes through all edges of $q$ that belong to $\EredT$ and since these edges are the only edges that are shared between $q$ and the children of $q$ in $\mathcal{T}^*$, Invariant~I.\ref{i:1} is satisfied. Invariants I.\ref{i:2} is satisfied, because edge $(u,v)$, which belongs to $\EredT$, is an internal chord of $C_q$, and all edges of $q$ that belong to $\EredT$ are on $C_q$. For Invariant~I.\ref{i:3}, we argue as follows. Consider an edge $e$ of $q$ that belongs to $\Eblue$. If both the (red) edges that precede and that follow $e$ in $q$ belong to $\EredT$, then $e$ is on $C_q$; otherwise, $e$ is not in the interior of $C_q$. Hence, Invariant~I.\ref{i:3} is satisfied. Invariant~I.\ref{i:4} is trivially satisfied, as $q$ belongs to $\Dgreen$ and therefore by Lemma~\ref{lem:barnetteprops}.\ref{prp:2} does not contain any edge of $\Egreen$. Finally, Invariant~I.\ref{i:5.ii} is satisfied, since the only vertices of $q$ that belong to $C_q$ are the endvertices of the edges of $q$ that belong to $\EredT$ (and thus not to $\EredN$). We conclude this case by mentioning that $C_q$ can be easily proven to be plane and that the fact that $C_q$ is simple can be proved symmetrically to the case in which $q$ belongs to $\Dblue$.     

The base of our recursive algorithm corresponds to the face $\rho \in \Dblue$ that is the root of $\mathcal{T}^*$. In this case, by setting $C_\rho$ to be the facial cycle $\rho$, we trivially satisfy all invariant properties of our algorithm; see Fig.~\ref{fig:sample} for an example.  

\begin{figure}[t]
	\centering
	\subfloat[\label{fig:sample-cycle}{}]{
	\includegraphics[page=4,scale=0.6]{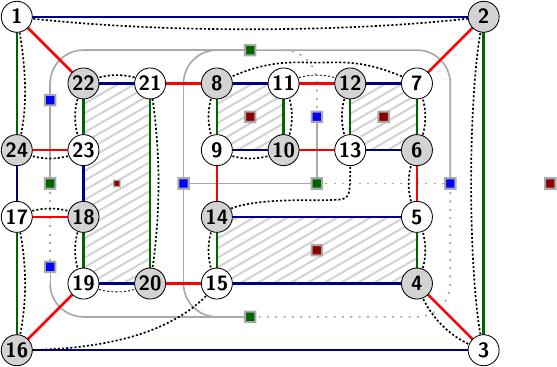}}
	\hfil
	\subfloat[\label{fig:sample-embed}{}]{
	\includegraphics[page=5,scale=0.6]{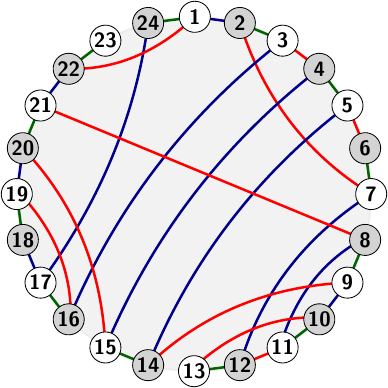}}
	\caption{Illustration of 
	(a)~a sample subhamiltonian cycle (dotted drawn), and 
	(b)~a corresponding dispersable book embedding with three pages produced by our algorithm.}
	\label{fig:sample}
\end{figure}

Once we traverse $\mathcal{T}^*$, we have computed a simple and plane cycle $C$, which by Invariants~\ref{i:2}-\ref{i:4}, satisfies Properties~\ref{prp:4}-\ref{prp:6} of our theorem. We show that $C$ is a subhamiltonian cycle of $G$ as follows. Since $\mathcal{T}^*$ is a spanning tree of $G_{bg}^*$, every green edge of $G$ bounds a face that is in $\mathcal{T}^*$, and by Invariant~\ref{i:4} we may assume that both its endpoints are consecutive along $C$. As every vertex is incident to a green edge, it follows that $C$ is indeed a subhamiltonian cycle of $G$.\qed
\end{proof}

Theorem~\ref{thm:barnette} implies that every Barnette graph admits a proper $3$-edge coloring and a two-page book embedding, in which the green edges connect consecutive vertices along the spine, while the red and blue edges must be on different pages. To achieve our initial goal, which was to show that every Barnette graph is dispersable, we only have to assign the green edges to an additional third page. This result is summarized in the following corollary.
 
\begin{corollary}
Every Barnette graph is dispersable.
\end{corollary} 

In relation to Barnette's conjecture, our result does not guarantee the existence of a Hamiltonian cycle. However, it guarantees the existence of a subhamiltonian cycles, which contains roughly $\frac{2n}{3}$ edges, where $n$ is the number of vertices of the graph. More precisely, all green edges, which are $\frac{n}{2}$ in total, are part of this cycle. Moreover, the red edges that belong to the set $\EredN$ in the proof of Theorem~\ref{thm:barnette} are by construction also part of this cycle. As a $3$-regular planar graph has exactly $\frac{n}{2}+2$ faces, we may assume w.l.o.g.\ that the number of nodes in $G_{bg}^*$ is at most $$\frac{2}{3}\left( \frac{n}{2}+2 \right) = \frac{1}{3} \left(n + 4 \right),$$ which implies that $$|\EredT| \leq \frac{1}{3} (n + 4) - 1.$$ As a result, the number of edges in $\EredN$ is at least $$\frac{n}{2} - \left( \frac{1}{3} (n + 4) - 1 \right) = \frac{n}{6} - \frac{1}{3}.$$ So, in total the number of green edges and the number of red edges in $\EredN$ is at least $\frac{n}{2} + \frac{n}{6} - \frac{1}{3} = \frac{2n}{3} - \frac{1}{3}$, as initially claimed.

\section{Conclusions}
\label{sec:conclusions}
In this paper, we studied dispersable book embeddings, and we demonstrated that for $k=3,4$, the dispersable book thickness of a $k$-regular bipartite graph is not necessarily $k$, thus disproving an old conjecture by Bernhart and Kainen~\cite{BK79}. 
There exists a number of interesting related questions raised by our work.
 
\begin{enumerate}[1.]
\item A natural question to ask is whether there exist a non-dispersable bipartite graph for every $k \geq 5$. Our
computational experiments confirm the hypothesis for $k=5,6$. For larger values of $k$, combinatorial proofs are needed. 
%In this direction and inspired from the construction of the Folkman graph, we managed to prove, using the SAT formulation presented in Section~\ref{sec:3-regular}, that this is indeed the case for the $6$-regular bipartite graph obtained from the complete graph $K_7$ by first subdividing each of its edges by a vertex and then by creating two more copies for each original vertex of $K_7$ with the same neighborhood.   
%
\item A common property of the Folkman graph and of the Gray graph is that both are not vertex-transitive. So, it is natural to ask whether vertex-transitive regular bipartite graphs are dispersable.
\item We conjecture that all (i.e., not necessarily $3$-connected) $3$-regular bipartite planar graphs are dispersable;
proving or disproving the conjecture is a possible future direction.    
\item Is it possible to guarantee an upper bound on the dispersable book thickness of $k$-regular bipartite graphs (e.g., $k+1$)?
\item What is the complexity of the testing problem whether a given graph is dispersable? The question is of interest also in the case of fixed linear order or fixed edge-to-page assignment.
\item More generally, it would be interesting to study the dispersable book thickness of (non-regular) bipartite planar graphs of maximum degree $\Delta(G)$. Extensive computational experiments indicate that all such graphs admit dispersable book
embeddings with $\Delta(G)$ pages. 
\end{enumerate}

% ============================================================

% ============================================================
%\clearpage
\bibliographystyle{splncs03}
\bibliography{references}

\clearpage
\appendix
\section*{\LARGE Appendix}

\section{Book Embeddings of the Folkman Graph}
\label{app:folkman}
\begin{figure}[h!]
	\centering
	\subfloat[\label{fig:folkman-book-embedding}{}]
	{\includegraphics[page=4,width=.85\textwidth]{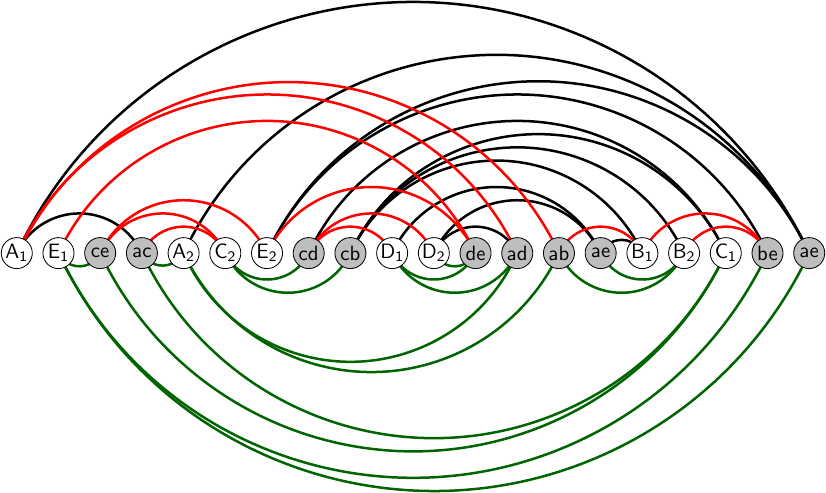}}
	\vfil
	\subfloat[\label{fig:folkman-dispersable}{}]
	{\includegraphics[page=5,width=.85\textwidth]{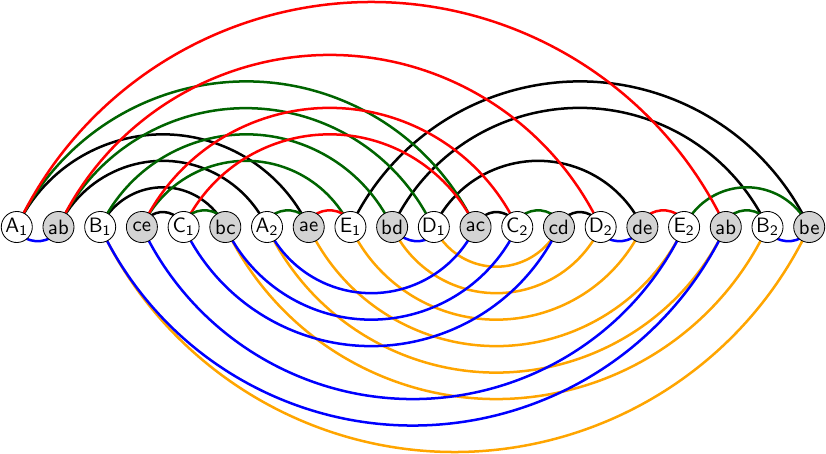}}
	\caption{%
	Different book embeddings of the Folkman graph:
	(a)~ordinary with $3$ pages, and
	(b)~dispersable with $5$ pages.}
	\label{fig:folkman-embeddings}
\end{figure}

\newpage

\section{Book Embeddings of the Gray Graph}
\label{app:gray}
\begin{figure}[h!]
	\centering
	\includegraphics[page=4,width=.85\textwidth]{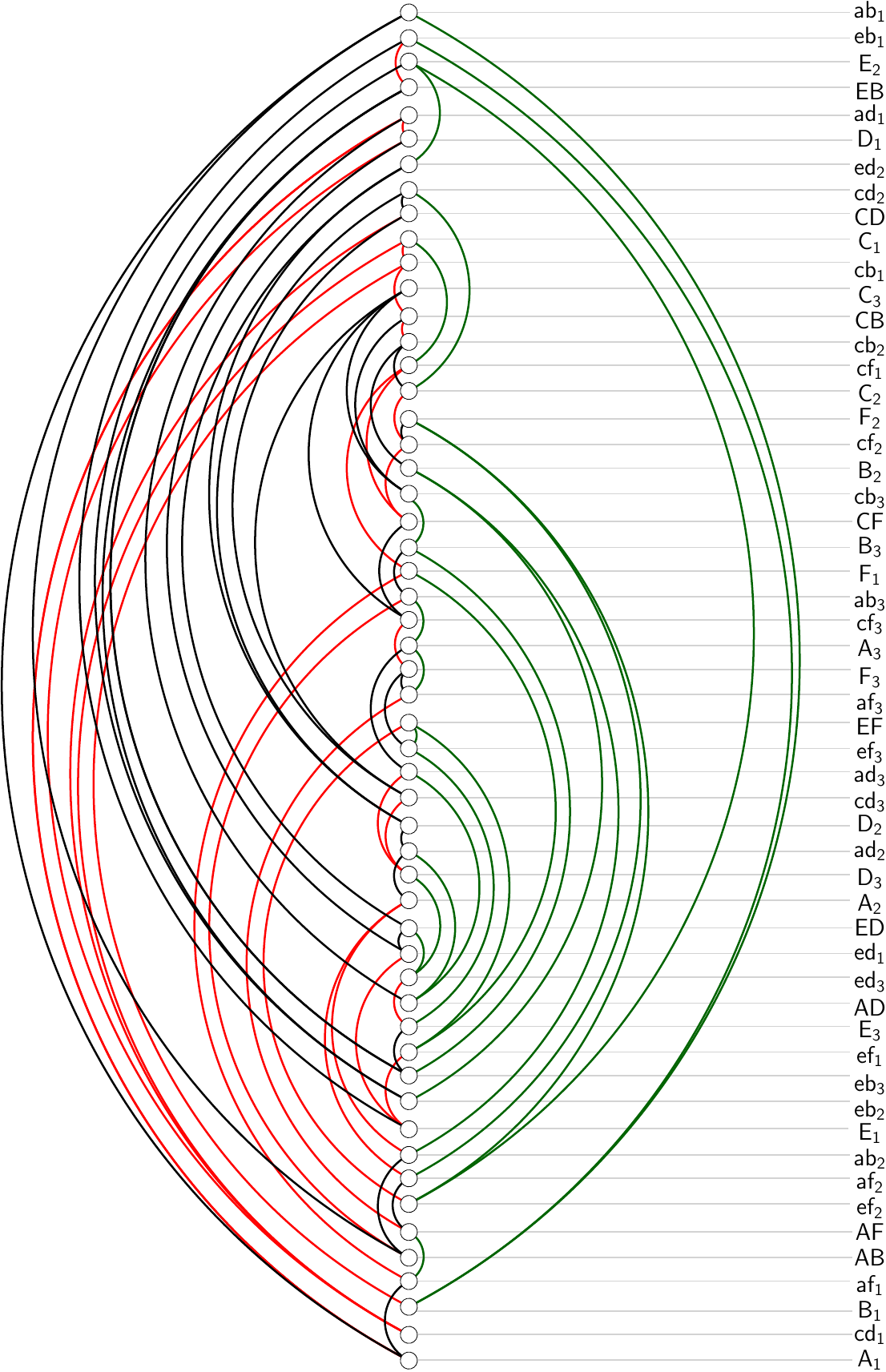}
	\caption{%
	An ordinary book embedding of the Gray graph with $3$ pages.}
	\label{fig:gray-book-embedding}
\end{figure}

\begin{figure}[p]
	\centering
	\includegraphics[page=5,width=.85\textwidth]{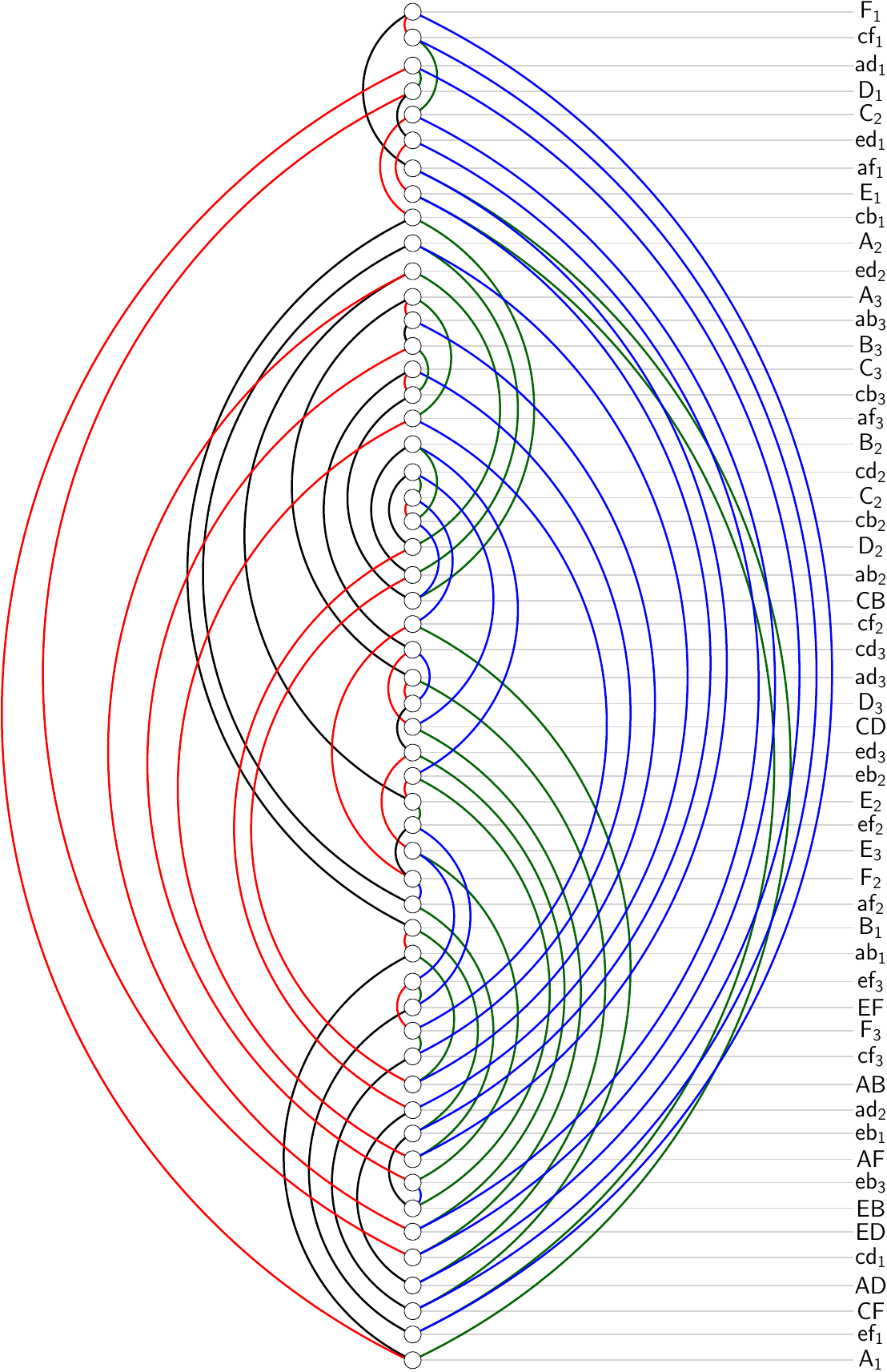}
	\caption{%
	A dispersable book embedding of the Gray graph with $4$ pages.}
	\label{fig:gray-dispersable}
\end{figure}

\end{document}